%% file: main.tex
\documentclass{article}
\pdfoutput=1

\usepackage[a4paper]{geometry}
\usepackage{float}
\usepackage{url}

\usepackage[pdftex]{graphicx}
\usepackage{listings,amsmath,amssymb,amsthm,amsfonts,tikz}
\usepackage[shortlabels]{enumitem}
\usepackage[section]{algorithm}
\usepackage[noend]{algorithmic}
\newcommand{\EndIf}{}

\newcommand\ttbraceleft{\texttt{\symbol{123}}}

\floatname{algorithm}{Procedure}

\usepackage{pgf,tikz}
\usepackage{mathrsfs}
\usetikzlibrary{arrows}
\usetikzlibrary{patterns}
\usetikzlibrary{backgrounds}

\usepackage{todonotes}

\input{macros}

\begin{document}
\title{Extending Two-Variable Logic on Trees
} 

\author{Bartosz Bednarczyk \and  Witold Charatonik\and Emanuel Kiero\'nski\\
Institute of Computer Science\\  University of Wroclaw, Poland \\ 
\texttt{bbednarczyk@stud.cs.uni.wroc.pl}\\ 
\texttt{\ttbraceleft wch,kiero\ttbraceright@cs.uni.wroc.pl}}

\newtheorem{definition}{Definition}
\newtheorem{example}{Example}
\newtheorem{theorem}{Theorem}
\newtheorem{lemma}{Lemma}
\newtheorem{corollary}{Corollary}
\newtheorem{proposition}{Proposition}

\maketitle

\begin{abstract} The finite satisfiability problem for the
  two-variable fragment of first-order logic interpreted over trees
  was recently shown to be \ExpSpace-complete. We consider two
  extensions of this logic. We show that adding either additional
  binary symbols or counting quantifiers to the logic does not affect
  the complexity of the finite satisfiability problem. However,
  combining the two extensions and adding both binary symbols and
  counting quantifiers leads to an explosion of this complexity.

  We also compare the expressive power of the two-variable fragment
  over trees with its extension with counting quantifiers. It turns
  out that the two logics are equally expressive, although counting
  quantifiers do add expressive power in the restricted case of
  unordered trees.
\end{abstract}

\begingroup
\let\clearpage\relax
\include{introduction}
\include{preliminaries}

\include{FO2Bin}
\include{C2}
\include{expressivepower}

\include{conclusion}
\endgroup

\bibliography{bibliography}

\end{document}

%% file: macros.tex
\renewcommand{\phi}{\varphi} % Nicer-looking phi

\newcommand{\AAA}{\mbox{\large \boldmath $\alpha$}}

\newcommand{\FOt}{\mbox{$\mathrm{FO}^2$}}

\newcommand{\Ct}{\mbox{$\mbox{\rm C}^2$}}

\newcommand{\ExpSpace}{\textsc{ExpSpace}}
\newcommand{\NExpTime}{\textsc{NExpTime}}

\newcommand{\AExpTime}{\textsc{AExpTime}}

\newcommand{\str}[1]{{\mathfrak{#1}}}
\newcommand{\restr}{\!\!\restriction\!\!}

\newcommand{\N}{{\mathbb N}}

\newcommand{\sss}{\scriptscriptstyle}

\newif\ifdraftpaper
\draftpapertrue
\ifdraftpaper
\newcommand{\nb}[1]{\textcolor{red}{\bf\large \#}\footnote{\textcolor{blue}{#1}}}
\else
\newcommand{\nb}[1]{}
\fi

\newcommand{\succv}{{\downarrow}}
\newcommand{\lessv}{{\downarrow^{\scriptscriptstyle +}}}
\newcommand{\succh}{{\rightarrow}}
\newcommand{\lessh}{{\rightarrow^{\scriptscriptstyle +}}}
\newcommand{\succlessv}{{\downarrow \downarrow^+}}
\newcommand{\predgreatv}{{\uparrow \uparrow^+}}

\newcommand{\tsuccv}{\theta_{\downarrow}}
\newcommand{\tprecv}{\theta_{\uparrow}}
\newcommand{\tlessv}{\theta_{\downarrow \downarrow^+}}
\newcommand{\tgreatv}{\theta_{\uparrow \uparrow^+}}
\newcommand{\tsucch}{\theta_{\rightarrow}}
\newcommand{\tprech}{\theta_{\leftarrow}}
\newcommand{\tlessh}{\theta_{\rightrightarrows^+}}
\newcommand{\tgreath}{\theta_{\leftleftarrows^+}}
\newcommand{\tfree}{\theta_{\not\sim}}
\newcommand{\teq}{\theta_{=}}

\newcommand{\FOtall}{\mbox{$\FOt[\succv, \lessv, \succh, \lessh, \tau_{com}]$}}
\newcommand{\Ctall}{\mbox{$\Ct[\succv, \lessv, \succh, \lessh, \tau_{com}]$}}

\newcommand{\CTwoFull}{\mbox{$\Ct[\succv, \lessv, \succh, \lessh]$}}

\newcommand{\Wprop}[3]{W \hspace*{-3pt}\left\langle  #1, #2, #3 \right\rangle}
\newcommand{\Psiprop}[3]{\Psi \langle #1,#2, #3 \rangle}

\newcommand{\posone}{
\scalebox{0.8}{
\begin{tikzpicture}[baseline=0.7ex,scale=1]
\node at (0,0) {$\bullet$};
\node at (0,0.45) {$\circ$};
\node at (0,0.45) {$\times$};
\node at (0.1,0.25) {$\lessv$};
\end{tikzpicture}
}
}

\newcommand{\postwo}{
\scalebox{0.8}{
\begin{tikzpicture}[baseline=0.7ex,scale=1]
\node at (0,0) {$\bullet$};
\node at (0,0.45) {$\bullet$};
\node at (0,0.45) {$\times$};
\node at (0.1,0.25) {$\lessv$};
\end{tikzpicture}
}
}

\newcommand{\posthree}{
\scalebox{0.8}{
\begin{tikzpicture}[baseline=0.7ex,scale=1]
\node at (0,0) {$\bullet$};
\node at (0,0.45) {$\bullet$};
\node at (0.1,0.25) {$\lessv$};
\node at (-0.25,0.45) {$\rightarrow$};
\node at (-0.1,0.6) {${\sss +}$};
\node at (-0.45,0.45) {$\circ$};
\node at (-0.45,0.45) {$\times$};
\end{tikzpicture}
}
}

\newcommand{\posfour}{
\scalebox{0.8}{
\begin{tikzpicture}[baseline=0.7ex,scale=1]
\node at (0,0) {$\bullet$};
\node at (0,0.45) {$\circ$};
\node at (0.1,0.25) {$\lessv$};
\node at (-0.25,0.45) {$\rightarrow$};
\node at (-0.1,0.6) {${\sss +}$};
\node at (-0.45,0.45) {$\circ$};
\node at (-0.45,0.45) {$\times$};
\end{tikzpicture}
}
}

\newcommand{\posfive}{
\scalebox{0.8}{
\begin{tikzpicture}[baseline=0.7ex,scale=1]
\node at (0,0) {$\bullet$};
\node at (0,0.45) {$\circ$};
\node at (0,0.45) {$\times$};
\node at (0,0.25) {$\succv$};
\end{tikzpicture}
}
}

\newcommand{\possix}{
\scalebox{0.8}{
\begin{tikzpicture}[baseline=0.7ex,scale=1]
\node at (0,0) {$\bullet$};
\node at (0,0.45) {$\circ$};
\node at (0,0.45) {$\times$};
\node at (0.1,0.25) {$\lessv$};
\node at (0.17,0.2) {${\sss +}$};
\end{tikzpicture}
}
}

\newcommand{\posseven}{
\scalebox{0.8}{
\begin{tikzpicture}[baseline=0.7ex,scale=1]
\node at (0,0) {$\circ$};
\node at (0,0.45) {$\bullet$};
\node at (0,0) {$\times$};
\node at (0.1,0.25) {$\lessv$};
\node at (0.17,0.2) {${\sss +}$};
\end{tikzpicture}
}
}

\newcommand{\poseight}{
\scalebox{0.8}{
\begin{tikzpicture}[baseline=2ex,scale=1]
\node at (0,0.45) {$\bullet$};
\node at (-0.25,0.45) {$\rightarrow$};
\node at (-0.1,0.6) {${\sss +}$};
\node at (-0.45,0.45) {$\circ$};
\node at (-0.45,0.45) {$\times$};
\end{tikzpicture}
}
}

\newcommand{\posnine}{
\scalebox{0.8}{
\begin{tikzpicture}[baseline=2ex,scale=1]
\node at (0,0.45) {$\circ$};
\node at (-0.25,0.45) {$\rightarrow$};
\node at (-0.1,0.6) {${\sss +}$};
\node at (-0.45,0.45) {$\bullet$};
\node at (0,0.45) {$\times$};
\end{tikzpicture}
}
}

\newcommand{\posten}{
\scalebox{0.8}{
\begin{tikzpicture}[baseline=2ex,scale=1]
\node at (0,0.45) {$\bullet$};
\node at (-0.25,0.45) {$\rightarrow$};
\node at (-0.1,0.6) {${\sss +}$};
\node at (-0.25,0.6) {${\sss +}$};
\node at (-0.45,0.45) {$\circ$};
\node at (-0.45,0.45) {$\times$};
\end{tikzpicture}
}
}

\newcommand{\poseleven}{
\scalebox{0.8}{
\begin{tikzpicture}[baseline=2ex,scale=1]
\node at (0,0.45) {$\circ$};
\node at (-0.25,0.45) {$\rightarrow$};
\node at (-0.1,0.6) {${\sss +}$};
\node at (-0.25,0.6) {${\sss +}$};
\node at (-0.45,0.45) {$\bullet$};
\node at (0,0.45) {$\times$};
\end{tikzpicture}
}
}

\newcommand{\postwelve}{
\scalebox{0.8}{
\begin{tikzpicture}[baseline=0.7ex,scale=1]
\node at (0,0) {$\bullet$};
\node at (0,0.45) {$\bullet$};
\node at (0.1,0.25) {$\lessv$};
\node at (-0.25,0.45) {$\rightarrow$};
\node at (-0.1,0.6) {${\sss +}$};
\node at (-0.45,0.45) {$\circ$};
\node at (-0.45,0) {$\times$};
\node at (-0.7,0.45) {$\rightarrow$};
\node at (-0.9,0.45) {$\bullet$};
\node at (-0.55,0.6) {${\sss +}$};
\node at (-0.9,0) {$\bullet$};
\node at (-0.8,0.25) {$\lessv$};
\node at (-0.45,0) {$\circ$};
\node at (-0.35,0.25) {$\lessv$};
\end{tikzpicture}
}
}

\newcommand{\posthirteen}{
\scalebox{0.8}{
\begin{tikzpicture}[baseline=0.7ex,scale=1]
\node at (0,0) {$\circ$};
\node at (0,0.45) {$\bullet$};
\node at (0,0) {$\times$};
\node at (0.1,0.25) {$\lessv$};
\end{tikzpicture}
}
}

\newcommand{\posfourteen}{
\scalebox{0.8}{
\begin{tikzpicture}[baseline=0.7ex,scale=1]
\node at (0,0) {$\bullet$};
\node at (0,0.45) {$\bullet$};
\node at (0.1,0.25) {$\lessv$};
\node at (0.25,0.45) {$\rightarrow$};
\node at (0.25,0.6) {${\sss +}$};
\node at (0.45,0.45) {$\circ$};
\node at (0.45,0.45) {$\times$};
\end{tikzpicture}
}
}

\newcommand{\posfifteen}{
\scalebox{0.8}{
\begin{tikzpicture}[baseline=0.7ex]
\node at (0,0) {$\bullet$};
\node at (0,0.45) {$\circ$};
\node at (0.1,0.25) {$\lessv$};
\node at (0.25,0.45) {$\rightarrow$};
\node at (0.25,0.6) {${\sss +}$};
\node at (0.45,0.45) {$\circ$};
\node at (0.45,0.45) {$\times$};
\end{tikzpicture}
}
}

\newcommand{\possixteen}{
\scalebox{0.8}{
\begin{tikzpicture}[baseline=0.7ex,scale=1]
\node at (0,0) {$\bullet$};
\node at (0,0.45) {$\bullet$};
\node at (0.1,0.25) {$\lessv$};
\node at (-0.25,0.45) {$\rightarrow$};
\node at (-0.1,0.6) {${\sss +}$};
\node at (-0.45,0.45) {$\circ$};
\node at (-0.45,0.45) {$\times$};
\node at (-0.7,0.45) {$\rightarrow$};
\node at (-0.9,0.45) {$\bullet$};
\node at (-0.55,0.6) {${\sss +}$};
\node at (-0.9,0) {$\bullet$};
\node at (-0.8,0.25) {$\lessv$};
\end{tikzpicture}
}
}

\newcommand{\cutout}[1]{}

\newcommand{\fw}{{\not\sim}}

\newcommand{\tp}[2]{{\rm tp}^{\str{#1}}({#2})}
\newcommand{\ftp}[3]{{\rm ftp}^{\str{#1}}_{#2}({#3})}
\newcommand{\rftpa}[2]{{\rm rftp}_{{#1}}({#2})}
\newcommand{\rftpv}[3]{{\rm rftp}^{\str{#1}}_{#2}({#3})}
\newcommand{\hftp}[3]{{\rm hftp}^{\str{#1}}_{#2}({#3})}
\newcommand{\hshift}{\hspace*{20pt}}

\newcommand{\wcf}[2]{W^{#1}_{#2}}
\newcommand{\wcfp}[4]{W^{#1}_{#2}(#3, #4)}

%% file: introduction.tex
\section{Introduction}
\paragraph*{Two-variable logics}
Two-variable logic, \FOt, is one of the most prominent decidable
fragments of first-order logic. It is important in computer science
because of its decidability and connections with other formalisms like
modal, temporal and description logics or query languages. For
example, it is known that \FOt{} over words can express the same
properties as unary temporal logic~\cite{EtessamiVW02} and \FOt{} over
trees is precisely as expressive as the navigational core of XPath,
a~query languages for XML documents~\cite{MarxDeRijke04}. The
complexity of the satisfiability problem for \FOt{} over words and
trees, respectively, is studied in \cite{EtessamiVW02}, and \cite{BBCKLMW-tocl}.
Namely, it is shown that its satisfiability problem over words is \NExpTime-complete and over trees---\ExpSpace-complete.

On the other hand, \FOt{} cannot express that a~structure is a~word or
a~tree and it cannot express that a~relation is transitive, an
equivalence or an order. This lead to extensive studies of \FOt{} over
various classes of structures, where some distinguished relational
symbols are interpreted in a~special way, e.g., as equivalences or
linear orders. The finite satisfiability problem for \FOt{} remains
decidable over structures where one \cite{KieronskiOttoLics05} or two
relation symbols \cite{KieronskiT09} are interpreted as equivalence
relations; where one
\cite{Otto2001} or two relations are interpreted as linear
orders~\cite{SchwentickZ10, Zeume16}; where two relations are interpreted as
successors of two linear orders~\cite{Manuel10, Figueira2012,
  CharatonikWitkowski-tocl}; where one relation is interpreted as
linear order, one as its successor and another one as equivalence
\cite{BojanczykDMSS11}; where one relation is
transitive~\cite{szwastT13}; where an equivalence closure can be
applied to two binary predicates~\cite{KieroMPT-lics12}; where
deterministic transitive closure can be applied to one binary
relation~\cite{CharatonikKM14}. It is known that the finite
satisfiability problem is undecidable for \FOt{} with two transitive
relations~\cite{Kieronski-csl12}, with three equivalence
relations~\cite{KieronskiOttoLics05}, with one transitive and one
equivalence relation~\cite{KieronskiT09}, with three linear
orders~\cite{Kieronski-csl11}, with two linear orders and their two
corresponding successors~\cite{Manuel10}. A summary of complexity
results for extensions of \FOt{} with binary predicates being the
order relations 
can be found in~\cite{Zeume16}.

In the context of extensions of \FOt{} it is enough to consider
relational signatures with symbols of arity at most
two~\cite{GrKoVa97}. Some of the above mentioned decidability results,
e.g.,~\cite{BBCKLMW-tocl, SchwentickZ10, Manuel10, Figueira2012,
  BojanczykDMSS11, CharatonikKM14}, are obtained under the restriction that besides the distinguished 
	binary symbols interpreted in a special way there
are no other binary predicates in the signature; some, like~\cite{
  KieronskiOttoLics05, KieronskiT09, Otto2001,
  CharatonikWitkowski-tocl, szwastT13, KieroMPT-lics12, Zeume16} are valid in
the general setting. 
In the undecidability results additional binary predicates are usually not  necessary.

Another decidable extension of \FOt{} is the two-variable fragment
with \emph{counting quantifiers}, \Ct, where quantifiers of the form
$\exists^{\leq k}$, $\exists^{=k}$, $\exists^{\geq k}$ are allowed.
The finite satisfiability problem for \Ct{} was proved to be decidable
and \NExpTime-complete (both under unary and binary encoding of
numbers in counting quantifiers) in~\cite{GradelOR97,
  PacholskiSTLics97, IPH05}. There are also decidable extensions of
\Ct{} with special interpretations of binary symbols:
in~\cite{CharatonikWitkowski-tocl} two relation symbols are
interpreted as child relations in two forests (which subsumes the case
of two successor relations on two linear orders), in
\cite{Pratt-lics14} one symbol is interpreted as equivalence relation
and in \cite{WchPwitC2LinearOrder} one symbol is interpreted as linear
order (and the case with two linear orders is undecidable).

\paragraph*{Our contribution}
In this paper we extend the main result from \cite{BBCKLMW-tocl},
namely \ExpSpace-completeness of the satisfiability problem for \FOt{}
interpreted over finite trees without additional binary symbols.  We consider
two extensions of this logic. We show that adding either additional
binary symbols or counting quantifiers to the logic does not increase
the complexity of the satisfiability problem. However, when we combine
the two extensions and add both binary symbols and counting
quantifiers then the complexity explodes and the problem is at
least as hard as the emptiness problem for vector addition tree
automata~\cite{deGrooteGS04}. Since emptiness of vector addition tree
automata is a~long-standing open problem, showing decidability of
\Ct{} over trees with additional binary symbols is rather unlikely in
nearest future.

Let us recall that the situation is similar to the case of finite words:
\FOt{} with a linear order and the induced successor relation remains
\NExpTime-complete when extended either with additional binary relations \cite{Zeume16}
or with counting quantifiers \cite{WchPwitC2LinearOrder}. Combining both additional 
ingredients gives a logic which this time is know to be decidable, but with very high complexity, as it is equivalent to 
the emptiness problem of multicounter automata \cite{WchPwitC2LinearOrder}.

We additionally compare the expressive power of the two-variable fragment over
trees with its extension with counting quantifiers. It is not
difficult to see that \FOt{} over unordered trees cannot count and
thus \Ct{} is strictly more expressive in this case. However, the
presence of order in form of sibling relations gives \FOt{} the
ability of counting and makes the two logics equally expressive.

%%% Local Variables: 
%%% mode: latex
%%% TeX-master: "../main"
%%% End: 

%% file: preliminaries.tex
\section{Preliminaries}

\subsection{Logics, trees and atomic types}
We work with signatures of the form $\tau=\tau_0 \cup \tau_{nav} \cup
\tau_{com}$, where $\tau_0$ is a set of unary symbols, $\tau_{nav}= \{
\succv, \lessv, \succh, \lessh \}$ is the set of \emph{navigational}
binary symbols, and $\tau_{com}$ is a set of \emph{common} binary
symbols.  Over such signatures we consider two fragments of
first-order logic: \FOt{}, i.e., the restriction of first-order logic
in which only variables $x$ and $y$ are available, and its extension
with \emph{counting quantifiers}, \Ct{}, in which quantifiers of the
form $\exists^{\ge n}$, $\exists^{\le n}$, for $n \in \N$ are allowed.
We assume that the reader is familiar with their standard semantics.

We write \FOt$[\tau_{bin}]$ or \Ct$[\tau_{bin}]$ where $\tau_{bin} \subseteq \tau_{nav} \cup \tau_{com}$ to denote that the only binary symbols that
are allowed in signatures are from $\tau_{bin}$. We will mostly work with two logics: \FOt{}$[\succv, \lessv, \succh, \lessh, \tau_{com}]$,
for $\tau_{com}$ being an arbitrary set of common binary symbols, and $\CTwoFull$, i.e., the fragment with counting quantifiers with
no common binary symbols.

We are interested in finite unranked, ordered tree structures, in which the interpretation of the symbols from $\tau_{nav}$ is fixed: 
$\succv$ is interpreted as the child relation, $\succh$ as the right sibling relation, and $\lessv$ and $\lessh$
as their respective transitive closures.  
We read $u \downarrow w$ as "$w$ is a \emph{child} of $u$" and $u \rightarrow w$ as "$w$ is the \emph{right sibling} of $u$". We will also use other standard terminology like \emph{ancestor}, \emph{descendant}, \emph{preceding-sibling}, \emph{following-sibling}, etc.

We use $x \fw y$ to abbreviate the formula stating that $x$ and $y$ are in \emph{free position},
i.e., that they are related by none of the navigational binary predicates available in the signature. 
Let us call the formulas specifying the relative position of a pair of elements in a tree with respect to binary navigational predicates \emph{order formulas}.
There are ten possible order formulas: 
$x \succv y$, $y \succv x$, $x \lessv y \wedge \neg (x \succv y)$, $y \lessv x \wedge \neg (y \succv x)$, $x \succh y$, $y \succh x$, $x \lessh y \wedge \neg (x \succh y)$, $y \lessh x \wedge \neg (y \succh x)$, $x \fw y$, $x{=}y$. 
They are denoted, respectively, as: $\tsuccv$, $\tprecv$, $\tlessv$, $\tgreatv$, 
$\tsucch$, $\tprech$, $\tlessh$, $\tgreath$, $\tfree$, $\teq$. Let $\Theta$ be the set of these ten formulas. 

We use symbol $\str{T}$ (possibly with sub- or superscripts) to denote tree structures. For a given
tree $\str{T}$ we denote by $T$ its universe. 
A \emph{tree frame} is a tree over a signature containing no unary predicates and no common binary predicates.
We will sometimes say that a tree frame $\str{T}_f$ is the tree frame \emph{of} $\str{T}$, or that $\str{T}$ is \emph{based} of $\str{T}_f$
if $\str{T}_f$ is obtained from $\str{T}$ by dropping the interpretation of all unary and common
binary symbols. We say that a formula $\phi$ is satisfiable \emph{over} a tree frame
if it has a model based on this tree frame.

Given a tree $\str{T}$, we say that a node $v \in T$ is a \emph{minimal} node (having some fixed property)
if there is no $w \in T$ (having this property)  such that $\str{T} \models w \lessv v$. 
A $\succv$-path ($\succh$-path) is a sequence of nodes $v_1, \ldots, v_k$ such that $\str{T} \models v_i \succv v_{i+1}$  ($\str{T} \models v_i \succh v_{i+1}$), for $i=1, \ldots, k-1$. Given a $\succv$-path ($\succh$-path) $P$
we say that distinct nodes $v_1, \ldots, v_l$ (having some fixed property) are $l$ \emph{smallest} elements (having this property) on $P$ if 
for any other $v \in P$ (having this property) we have $\str{T} \models v_i \lessv v$ ($\str{T} \models v_i \lessh v$) for $i=1, \ldots, l$.
Analogously we define \emph{maximal} and \emph{biggest} elements.

An (atomic) \emph{$1$-type} is a maximal satisfiable set of atoms or negated atoms with free variable $x$. Similarly,
an (atomic) \emph{$2$-type} is a maximal satisfiable set of atoms or negated atoms with free variables $x,y$. 
Note that the numbers of atomic $1$- and $2$-types are bounded exponentially in the size of the signature. 
We often identify a 
type with the conjunction of all its elements. If we work with a signature with empty $\tau_{com}$ then $1$-types
correspond to subsets of $\tau_0$.
We denote by $\AAA_\varphi$ the set of $1$-types over the signature consisting of symbols appearing in $\varphi$.

For a given $\tau$-tree $\str{T}$, and a node $v \in T$ we say that $v$
\emph{realizes} a $1$-type $\alpha$ if $\alpha$ is the unique $1$-type
such that $\str{T} \models \alpha[v]$.  We denote by $\tp{T}{v}$ 
the $1$-type realized by $v$. Similarly, for
distinct $u,v \in T$, we denote by $\tp{T}{u,v}$ the unique
$2$-type \emph{realized} by the pair $u,v$, i.e.~the type $\beta$ such
that $\str{T} \models \beta[u,v]$.

\subsection{Normal forms}

As usual when working with satisfiability of two-variable logics we employ Scott-type normal form. 
We start with its adaptation for the case of \FOtall{}.

\begin{definition}\label{def:fo2normalform}
  We say that an \FOtall{} formula $\phi$ is in \emph{normal form} if
$$\phi=\forall xy \chi(x,y) \wedge \bigwedge\limits_{i=1}^{m} \forall x (\lambda_i(x) \Rightarrow \exists y (\theta_i(x,y) \wedge \chi_i(x,y))),$$
where $\lambda_i(x)$ is an atomic formula $A(x)$ for some unary symbol
$A$, $\chi(x,y)$ and $\chi_i(x,y)$ are quantifier-free, and
$\theta_i(x,y)$ is an order formula.
\end{definition}
	
Please note that the equality symbol may be used in $\chi$, e.g., we
can enforce that a model contains at most one node satisfying $A$:
$\forall xy (A(x) \wedge A(y) \Rightarrow x{=}y)$.
The following lemma can be proved in a standard fashion (cf.~e.g.,
\cite{BBCKLMW-tocl}).

\begin{lemma} \label{l:normalformbin} Let $\phi$ be an \FOtall{}
  formula over a signature $\tau$. There exists a polynomially
  computable \FOtall{} normal form formula $\phi'$ over signature
  $\tau'$ consisting of $\tau$ and some additional unary symbols, such
  that $\phi$ and $\phi'$ are satisfiable over the same tree frames.
\end{lemma}

Consider a conjunct $\phi_i=\forall x (\lambda_i(x) \Rightarrow
\exists y (\theta_i(x,y) \wedge \chi_i(x,y)))$ of an \FOtall{} normal
form formula $\phi$.  Let $\str{T} \models \phi$, and let $v \in T$ be an
element such that $\str{T} \models \lambda_i[v]$. Then an element $w
\in T$ such that $\str{T} \models \theta_i[v,w] \wedge \chi_i[v,w]$ is
called a \emph{witness} for $v$ and $\phi_i$.  We call $w$ an
\emph{upper} witness if $\theta_i(v, w) \models w \lessv v$, a
\emph{lower} witness if $\theta_i(v, w) \models v \lessv w$, a sibling
witness if $\theta_i (v, w) \models v \lessh w \vee w \lessh v$, and a
\emph{free} witness if $\theta_i(v, w) \models v \fw w$. We also
sometimes simply speak about $\lessh$-witnesses, $\uparrow$-witnesses,
etc.

For \Ct{} we use a similar but slightly different normal form. One
obvious difference is that it uses counting quantifiers, the other is
that its $\forall \exists$-conjuncts does not need to contain the
$\theta_i$-components, specifying the position of the required
witnesses.  Refining the normal form by incorporating those components
is possible but seems to require an exponential blow-up.

\begin{definition} \label{def:c2normalform} We say that a formula
  $\varphi \in \CTwoFull$ is in \textit{normal form}, if: $$\varphi =
  \forall x \forall y \ \chi(x,y) \wedge \bigwedge_{i=1}^{m} \left(
    \forall x \ \exists^{\bowtie_i C_i} y \ \chi_i(x,y) \right),$$
  where $\bowtie_i \in \lbrace \leq, \geq \rbrace$, each $C_i$ is a
  natural number, and $\chi(x,y)$ and all $\chi_i(x,y)$ are
  quantifier-free.
\end{definition}

\begin{lemma}[\cite{GradelOR97}]\label{lemma:normalform}
  Let $\varphi$ be a formula from $\CTwoFull$ over a
  signature~$\tau$. There exists a polynomially computable $\CTwoFull$
  formula $\varphi'$ over signature $\tau'$ consisting of $\tau$ and
  some additional unary symbols, such that $\varphi$ and $\varphi'$
  are satisfiable over the same tree frames.
\end{lemma}

As in the case of \FOtall{} we speak about \emph{witnesses}. Given a
normal for $\CTwoFull$ formula $\varphi$ and a tree $\str{T} \models
\varphi$, we say that a node $w \in T$ is a witness for $v \in T$ and
a conjunct $\forall x \ \exists^{\bowtie_i C_i} y \ \chi_i(x,y)$ of
$\varphi$ if $\str{T} \models \chi_i[v,w]$. If additionally $\str{T}
\models w \lessv v$ then $w$ is an \emph{upper} witness, if $\str{T}
\models v \lessv w$ then $w$ is a \emph{lower} witness, and so on.

In Section \ref{sec:f2bin}, when a normal form formula $\phi$ is considered
we always assume that it is as in Definition \ref{def:fo2normalform}. In 
particular we allow ourselves, wihtout explicitly recalling the shape of $\phi$, to refer to its  parameter $m$ and components
$\chi, \chi_i, \lambda_i, \theta_i$. Analogously, in Section \ref{sec:c2} we 
assume that any normal form $\phi$ is as in Definition \ref{def:c2normalform}.

%%% Local Variables: 
%%% mode: latex
%%% TeX-master: "../main"
%%% End: 

%% file: FO2Bin.tex
\section{\texorpdfstring{$\FOt$}{F2} on trees with additional binary relations}
\label{sec:f2bin}

In this section we show that the complexity of the satisfiability
problem for \FOt$[\succv, \lessv, \succh, \lessh]$ \cite{BBCKLMW-tocl} is
retained when the logic is extended with additional, uninterpreted
binary relations.

\begin{theorem} \label{t:fotbin}
The satisfiability problem for  \FOtall{} over finite trees is \ExpSpace-complete.
\end{theorem}

The lower bound is inherited from \FOt$[\succv, \lessv, \succh,
\lessh]$. For the upper bound we show that any satisfiable formula
$\phi$ has a model of depth and degree bounded exponentially in
$|\phi|$. Then we show an auxiliary result allowing us to restrict
attention to models in which all elements have free witnesses in a
relatively small fragment of the tree. We finally design an
alternating exponential time procedure searching for such small
models.

\subsection{Small model property}

Let $\mathfrak{f}$ be a fixed function, which for a given normal form
\FOtall{} formula $\varphi$ returns $96m^3|\AAA_\varphi|^3$. Recall
that $m$ is the number of $\forall \exists$-conjuncts of $\varphi$ and
$\AAA_\varphi$ is the set of $1$-types over the signature of $\phi$.
We will use $\mathfrak{f}$ to estimate the length of paths and the degree
of nodes in models. Note that for a given $\phi$ the value returned by
$\mathfrak{f}$ is exponentially bounded in $|\varphi|$.  It should be
mentioned that by a more careful analysis one could obtain slightly
better bounds (still exponential in $|\varphi|$), but $\mathfrak{f}$ is
sufficient for our purposes and allows for a reasonably simple
presentation.

The following small model property is crucial for obtaining
\ExpSpace-upper bound on the complexity of the satisfiability
problem. It can be seen as an extension of Theorem 3.3 from \cite{BBCKLMW-tocl},
where a similar result was proved for \FOt{} over trees without additional binary
relations.

\begin{theorem}[Small model theorem]\label{thm:fo2binsmallmodeltheorem}
  Let $\varphi$ be a satisfiable normal form \FOtall{} formula. Then
  $\varphi$ has a model in which the length of every $\succv$-path and
  the degree of each node are bounded exponentially in $|\varphi|$ by
  $\mathfrak{f}(\varphi)$.
\end{theorem}

We split the proof of this theorem into two lemmas. In the first one
we show how to shorten the $\succv$-paths and in the second --- how to
reduce the degree of nodes, i.e., to shorten $\succh$-paths.

\begin{lemma} \label{lemma:fo2binsmallpaths} Let $\varphi$ be a normal
  form \FOtall{} formula and $\str{T}$ its model. Then there exists a
  tree model $\str{T}'$ for $\varphi$ whose every $\downarrow$-path
  has length at most $\mathfrak{f}(\varphi)$.
\end{lemma}

\begin{proof}
  Assume that $\str{T}$ contains a $\downarrow$-path $P = \left( v_1,
    v_2, \ldots, v_n \right)$ longer than $\mathfrak{f}(\varphi)$. We
  show that then it is possible to remove some nodes from $\str{T}$
  and obtain a smaller model $\str{T}_0$. For a node $u \in T$ we
  define its \textit{projection onto} $P$ as the smallest node $v \in
  P$, such that $\str{T} \models v \downarrow^+ u$.

  We first distinguish a set $W$ of some relevant elements of
  $\str{T}$. $W$ will consist of four disjoint sets $W_0$, $W_1$,
  $W_2$, $W_3$.  For each $1$-type $\alpha$ we mark:
\begin{itemize}
\item $m$ biggest and $m$ smallest realizations of $\alpha$ on $P$ (or
  all realizations of $\alpha$ on $P$ if there are less than $m$ of
  them)
\item $m$ realizations of $\alpha$ outside $P$ having biggest
  projections onto $P$ and $m$ realizations of $\alpha$ outside $P$
  having smallest projections onto $P$ (or all realizations of
  $\alpha$ outside $P$ if there are less than $m$ of them).
\end{itemize}
Let $W_0$ be the set consisting of all the marked elements.  Let $W_1$
be a minimal (in the sense of $\subseteq$) set of nodes of $\str{T}$
such that all the elements from $W_0$ have all the required witnesses
in $W_0 \cup W_1$.  Similarly, let $W_2$ be a minimal set of nodes of
$\str{T}$ such that all the elements from $W_1$ have all the required
witnesses in $W_0 \cup W_1 \cup W_2$.  Finally, let $W_3$ be the set
of those projections onto $P$ of elements of $W_0 \cup W_1 \cup W_2$
which are not in $W_0 \cup W_1 \cup W_2$.  Let $W := W_0 \cup W_1 \cup
W_2 \cup W_3$.  To estimate the size of $W$, observe that $|W_0| \leq
4 m |\AAA_\varphi|$, $|W_1| \leq m |W_0|$, $|W_2| \leq m|W_1|$ and
$|W_3| \le |W_0 \cup W_1 \cup W_2|$.  Thus $|W| \le 24 m^3
|\AAA_\varphi|$.

An \emph{interval} of $P$ \emph{of length} $s$ is a sequence of nodes
of the form $(v_i, v_{i+1}, \ldots, v_{i+s-1})$ for some $i,s$.  We
claim that $P$ contains an interval $I$ of length at least
$2|\AAA_\varphi|^2+2$ having no elements in $W$. To the contrary
assume that there there is no such interval. Note that the extremal
points of $P$ (which are the root and a leaf of $\str{T}$) are members
of $W$. Hence the points of $W \cap P$ determine at most $|W|-1$
maximal (possibly empty) intervals not containing elements of $W$. It
follows that $|P| \le (|W|-1)(2 |\AAA_\varphi|^2 + 1) +
|W|<|W|(2|\AAA_\varphi|^2+2)$, which by simple estimations gives $|P|
< 96m^3|\AAA_\varphi|^3$, a contradiction.

Using the pigeonhole principle we can easily see that in $I$ there are
two disjoint pairs of nodes $v_k, v_{k+1}$ and $v_l,
v_{l+1}$, for some $k < l$ such that
$\tp{T}{v_{l+i}}=\tp{T}{v_{k+i}}$, for $i=0,1$. We build
a tree $\str{T}_0$ by replacing in $\str{T}$ the subtree rooted at
$v_{k+1}$ by the subtree rooted at $v_{l+1}$, setting
$\tp{T_{\rm 0}}{v_k, v_{l+1}}:=\tp{T}{v_k, v_{k+1}}$
and for each  $v$ being a sibling of $v_{k+1}$ in $\str{T}$ setting $\tp{T_{\rm 0}}{v, v_{l+1}}:=\tp{T}{v, v_{k+1}}$ (all the
remaining $2$-types are retained from $\str{T}$).  In effect,
all the subtrees rooted at elements of $P$ between $v_{k+1}$ and $\ldots, v_l$ are removed from $\str{T}$.  Please note that all
elements of $W$ survive our surgery. This guarantees that the elements
of $W_0 \cup W_1$ retain all their witnesses. However, some nodes $v$
from $T_0 \setminus (W_0 \cup W_1)$ could lose their witnesses. We can
now reconstruct them using the nodes from $W_0$.  Let us describe this
procedure, distinguishing several cases.

{\em Case 1:} $v=v_k$. All the siblings, ancestors and elements in
free position to $v_k$ from $\str{T}$ are retained in
$\str{T}_0$. Thus $v_k$ retains all its sibling, ancestor and free
witnesses. There is also no problem with $\succv$-witnesses, as $v_k$
retains all its children except $v_{k+1}$, and $v_{k+1}$ is replaced
by $v_{l+1}$ having the same $1$-type and connected to $v_k$ exactly
as $v_{k+1}$ was. Some $\succlessv$-witnesses for $v_k$ could be lost
however.  Let $B$ be a minimal (in the sense of $\subseteq$) set of
elements providing the required $\succlessv$-witnesses for $v_k$ in
$\str{T}$. Note that $|B| \le m$. Let $\alpha$ be a $1$-type realized
in $B$.  If all elements of $1$-type $\alpha$ from $B$ are in $W_0$
then there is nothing to do: they survive, and serve as proper
$\succlessv$-witnesses for $v_k$ in $\str{T}_0$. Otherwise, there must
be at least $m$ realizations of $\alpha$ in $W_0$ (on $P$ or outside
$P$) whose projections onto $P$ in $\str{T}$ are below $v_{l+2}$. We
can modify the $2$-types joining $v_k$ with some of them securing the
required $\succlessv$-witnesses for $v_k$. This can be done without
conflicts, since $v_k \not\in W_0 \cup W_1$ and hence it is not
required as a witness by any element of $W_0$.

{\em Case 2:} $v=v_{l+1}$.  All the 
descendants of $v_{l+1}$ are retained in $\str{T}_0$. 
Thus $v_{l+1}$ retains its descendant witnesses.  There 
is no problem with sibling witnesses since $v_{l+1}$ has
the same $1$-type as $v_{k+1}$ and it is connected to
its siblings in $\str{T}_0$ exactly as $v_{k+1}$ was in $\str{T}$. 
Using arguments
similar to these from the previous case we can show that also there is
no problem with upper witnesses for $v_{l+1}$. The only missing part
is to ensure that $v_{l+1}$ has all of its required free
witnesses. Let $B$ be a minimal (in the sense of $\subseteq$) set of
free witnesses for $v_{l+1}$ in $\str{T}$ and let $\alpha$ be a
$1$-type realized in $B$. If all elements of $1$-type $\alpha$ from
$B$ are in $W_0$ then there is nothing to do.

Otherwise, $v_{l+1}$ can reconstruct its witnesses from $B$ using 
$m$ realizations of $\alpha$ in $W_0$ outside $P$ with smallest projections onto $P$.
Note that they are indeed in free position to $v_{l+1}$ (since not all elements of $B$ are 
in $W_0$ and thus at least $m$ elements of $1$-type $\alpha$ from $W_0$ have
projections onto $P$ which are smaller than $v_k$).

{\em Case 3:} $v$ is a descendant of $v_{l+1}$. In this case $v_{l+1}$
retains all its sibling, descendant, and $\uparrow$-witnesses from
$\str{T}$.  Regarding $\predgreatv$-witnesses, consider the witnesses
of $1$-type $\alpha$ in $\str{T}$; either all of them are in $W_0$, or
they can be reconstructed using $m$ smallest realizations of $\alpha$
on $P$, which must be members of $W_0$. Regarding the free witnesses,
similarly, consider the witnesses of $1$-type $\alpha$ in $\str{T}$;
if not all of them are in $W_0$, then $v_{l+1}$ can reconstruct them
using $m$ elements of $1$-type $\alpha$ from $W_0$ outside $P$ with
smallest projections on $P$.

{\em Case 4:} $v$ is a child of $v_k$ different from $v_{l+1}$. Upper and lower
witnesses for $v$ are
retained in $\str{T}_0$. There is also no problem with sibling witnesses: even
if $v$ required $v_{k+1}$ as a witness in $\str{T}$ it can now use $v_{l+1}$.
Consider the case of free witnesses. Let $B$ as a minimal set of free
witnesses for $v$ in $\str{T}$ and let $C \subseteq B$ be the subset
of $B$ containing all the vertices from $B$ which lie inside the
subtree rooted at $v_{k+1}$. Observe that all the vertices from $B
\setminus C$ survive our surgery, so they can still serve as  proper
free witnesses for $v$. On the other hand, some vertices from $C$
could be lost. Consider the witnesses of $1$-type $\alpha$ in
$C$. if not all of them are in $W_0$, then there must be at least
$m$ realizations of $\alpha$ in $W_0$ in free position to $v$:
these are either biggest realizations of $\alpha$ on $P$ or realizations
of $\alpha$ with biggest projections onto $P$.  Thus $v$ can use
them to reconstruct its witnesses.

{\em Case 5:} $v$ is a descendant of a child of $v_k$ but not of
$v_{l+1}$. Observe that all of the required witnesses for $v$ except
the free witnesses are retained in $\str{T}_0$. To reconstruct the
free witnesses for $v$ we can use the strategy described in Case 4.

{\em Case 6:} $v$ is an ancestor of $v_k$. In this case $v$ retains
all its sibling, upper and free witnesses from $\str{T}$. To deal
with the lower witnesses we can simply follow the strategy from Case 1.

{\em Case 7:} $v$ is in free position to $v_k$. Note that all of the
witnesses for $v$ except free ones survived the surgery. It's possible
that some of the free witnesses for $v$ were lost, but we find the new
free witnesses exactly as in Case 4.

After the described adjustments all the elements of $\str{T}_0$ have
appropriate witnesses. Since all the $2$-types realized in $\str{T}_0$
are also realized in $\str{T}$ this ensures that the $\forall \forall$
conjunct of $\varphi$ is not violated in $\str{T}_0$. Thus $\str{T}_0
\models \varphi$.

Note that the number of nodes of $\str{T}_0$ is strictly smaller than
the number of nodes of $\str{T}$. We can repeat the same shrinking
process starting from $\str{T}_0$, and continue it, obtaining
eventually a model $\str{T}'$ whose paths are bounded as required.
\end{proof}

\begin{lemma} \label{lemma:fo2binsmalldegfinitetree} Let $\varphi$ be
  a normal form \FOtall{} formula and $\str{T} \models \varphi$. Then
  there exists a model $\str{T}' \models \varphi$, obtained by
  removing some subtrees from $\str{T}$ such that the degree of its
  every node is bounded by $\mathfrak{f}(\phi)$.
\end{lemma}

\begin{proof}
  Assume that $\str{T}$ contains a node $v$ having more than
  $\mathfrak{f}(\varphi)$ children.  We show that then it is possible
  to remove some of these children together with the subtrees rooted
  at them and obtain a smaller model $\str{T}' \models \varphi$.  The
  process is similar to the one described in the proof of Lemma
  \ref{lemma:fo2binsmallpaths}.  Let $P=(v_1, \ldots, v_k)$ be the
  $\succh$-path in $\str{T}$ consisting of all the children of $v$.
  We first distinguish a set $W$ of some relevant elements of
  $\str{T}$. It will consist of four disjoint sets $W_0$, $W_1$,
  $W_2$, $W_3$.

  For each $1$-type $\alpha$ we mark $m$ biggest and $m$ smallest
  realizations of $\alpha$ on $P$ (or all realizations of $\alpha$ on
  $P$ if there are less than $m$ of them). Further we choose $m+1$
  elements of $P$ having a realization of $\alpha$ as a descendant (or
  all such elements if there are less than $m+1$ of them) and for each
  of them mark one descendant of $1$-type $\alpha$.  Let $W_0$ be the
  set consisting of all the marked elements.  Let $W_1$ be a minimal
  set of nodes such that all the elements from $W_0$ have all the
  required witnesses in $W_0 \cup W_1$.  Similarly, let $W_2$ be a
  minimal set of nodes such that all the elements from $W_1$ have all
  the required witnesses in $W_0 \cup W_1 \cup W_2$.  Finally, let
  $W_3$ be the set of those elements of $P$ which are not in $W_0 \cup
  W_1 \cup W_2$ but have an element from $W_0 \cup W_1 \cup W_2$ in
  their subtree.  Let $W := W_0 \cup W_1 \cup W_2 \cup W_3$.  To
  estimate the size of $W$, observe that $|W_0| \leq (3m + 1)
  |\AAA_\varphi|$, $|W_1| \leq m |W_0|$, $|W_2| \leq m|W_1|$ $|W_3|
  \le |W_0 \cup W_1 \cup W_2|$.  Thus, after simple estimations, we
  have $|W| \le 24 m^3 |\AAA_\varphi|$.

An \emph{interval} of $P$ \emph{of length} $s$ is a sequence of nodes of the form $(v_i, v_{i+1}, \ldots, v_{i+s-1})$ for some $i,s$.
Using arguments similar to those from the proof of Lemma \ref{lemma:fo2binsmallpaths} we can show that $P$ contains an interval
$I$ with no elements in $W$, in which there are two disjoint pairs of nodes $v_k, v_{k+1}$ and $v_l, v_{l+1}$, for some $k < l$ such
that $\tp{T}{v_{l+i}}=\tp{T}{v_{k+i}}$, for $i=0,1$. 
We build an auxiliary tree  $\str{T}_0$ by removing the subtrees rooted at $v_{k+1}, \ldots, v_l$ and setting $\tp{T_{\rm 0}}{v_k, v_{l+1}}:=\tp{T}{v_k, v_{k+1}}$ (all the remaining $2$-types are retained from $\str{T}$). 
Again the elements which lost their witnesses in our construction can regain them by changing their connections to elements from $W_0$.
We explain that it can be done for all elements $v$ of $\str{T}_0$ distinguishing several cases.

{\em Case 1:} $v$ lies on path $P$ (for example $v = v_k$ or $v = v_{l+1}$).
Observe that the descendants and the ancestors of $v$
survive our surgery. Also, there is no problem with $\leftarrow$ and
$\rightarrow$ witnesses for any vertex $v$ on $P$ other than $v_k$ and
$v_{l+1}$. For $v_{k}$ and $v_{l+1}$ we simply observe that in
$\str{T}_0$ the right sibling of $v_{k}$ was replaced by the node with
exactly the same $1$-type as $v_{k+1}$ in $\str{T}$. The case of
$v_{l+1}$ is symmetric. Consider now the case of 
$\leftleftarrows^+$ witnesses (the case of  $\rightrightarrows^+$ witnesses is symmetric).
Let $B$ be a minimal (in the sense of $\subseteq$) set of
elements providing the required $\rightrightarrows^+$-witnesses for
$v$ in $\str{T}$. Note that $|B| \le m$. Let $\alpha$ be a $1$-type
realized in $B$.  If all elements of $1$-type $\alpha$ from $B$ are in
$W_0$ then there is nothing to do -- they survive, and serve as proper
$\rightrightarrows^+$-witnesses for $v$ in $\str{T}_0$. Otherwise,
there must be at least $m$ maximal realizations of $\alpha$ on $P$ to
the right of $v$. We can modify the $2$-types joining $v$ with some of
them securing the required $\rightrightarrows^+$-witnesses for
$v$. This can be done without conflicts, since $v$ requires at most
$m$ $\rightrightarrows^+$-witnesses, and $v \not\in W_0 \cup W_1$ and
hence it is not required as a witness by any element of
$W_0$. Finally, we need to show that $v$ has all required free
witnesses in $\str{T}_0$. And again, we consider a set $B$ of
all necessary free witnesses for $v$ in $\str{T}$ and take a
$1$-type $\alpha$ realized in $B$. If all $\alpha$-witnesses are in
$W_0$, there is nothing to do. Otherwise there are at least $m$
realizations of $\alpha$ in $W_0$, since we marked $m+1$ deep
realizations of $\alpha$ in different subtrees rooted at nodes from
$P$. By the fact that $v \not\in W_0 \cup W_1$ the vertex $v$ is not
required as a witness for $W_0$, so we can again modify the $2$-types
of these vertices to secure the required free witnesses for $v$.

{\em Case 2:} $v$ is an ancestor of $v_k$.
In this case all the required witnesses for $v$ other than its
descendants are retained in $\str{T}_0$. Regarding
$\succlessv$-witnesses, consider the witnesses of $1$-type $\alpha$
in $\str{T}$; either all of them are in $W_0$, or they can be
reconstructed using $m$ deep realizations of $\alpha$ below  path
$P$, which must be members of $W_0$. 

{\em Case 3:} $v$ is a descendant of a vertex from path $P$. 
All the descendants, siblings and ancestors of $v$ survive
the surgery.  To ensure that $v$ has the
required free witnesses we  follow the last part of the proof of
Case $1$.

{\em Case 4:} $v$ is in free position to of $v_k$. Again, only
free witnesses could be lost but they can be reconstructed as in
the previous cases.

And again, as in the proof of Lemma \ref{lemma:fo2binsmallpaths}, the process can be continued until a model with appropriately bounded degree of nodes is obtained.
\end{proof}

\subsection{Global free witnesses}

The small model property from the previous subsection is a crucial
step towards an exponential space algorithm for satisfiability. Note
however that it allows for models having doubly exponentially many
nodes, which thus cannot be stored in memory.  In the case of \FOt{}
without additional binary relations \cite{BBCKLMW-tocl} the corresponding
algorithm traversed $\succv$-paths guessing for each node $v$ its
\emph{full type} storing the sets of $1$-types of elements above,
below, and in free position to $v$, similarly to the case of \FOt{}
with counting from Section \ref{sec:c2}. Then it took care of
\emph{realizing} such full types. This approach would not be
sufficient for our current purposes, since the presence of additional
binary relations requires us not only to ensure that appropriate
$1$-types of elements will appear above, below and in free position to
a node but also that appropriate $2$-types will be realized. This is
especially awkward when dealing with free witnesses, since for a given
node they are located on different paths.  To overcome this difficulty
we show that we always can assume that all elements have their free
witnesses in small, exponentially bounded fragment of a model.

\begin{lemma}\label{l:fo2binglobalfree}
Let $\varphi$ be a normal form \FOtall{} formula and $\str{T}$ its model. Let $h$ be the length of the longest $\succv$-path in $\str{T}$ and $d$ the 
maximal number of $\succv$-successors of a node. Then there exists a tree $\str{T}'$ and a set of nodes $F \subseteq T'$, called a \emph{global set of free
witnesses} such that:
\begin{itemize}\itemsep0pt
\item the universes, the $1$-types of all elements and the tree frames of $\str{T}$ and $\str{T}'$ are identical,
\item $\str{T}' \models \varphi$,
\item the size of $F$ is bounded by $3( m + 1)^3 h^2 d^2 |\AAA_\phi|$,
\item $F$ is closed under $\uparrow$, $\leftarrow$ and $\rightarrow$,
\item for each conjunct of $\varphi$ of the form $\phi_i =\forall x (\lambda_i(x) \rightarrow  \exists y (x \fw y \wedge \chi(x,y)))$ and each node $v \in T'$,
if $\str{T}' \models \lambda_i[v]$ then 
there is a witness for $v$ and $\phi_i$ in $F$.
\end{itemize}
\end{lemma}

\begin{proof}
We say that an element $v$ is a \emph{minimal element of type} $\alpha$ in $\str{T}$ if $\tp{T}{v}=\alpha$ and there is no $w \in T$
such that $\tp{T}{w}=\alpha$ and $\str{T} \models w \lessv v$. 

We first describe a procedure which distinguishes in $\str{T}$ the
desired set $F$. This will contain three disjoint subsets $F_0, F_1,
F_2$. Start with $F_0=F_1=F_2=\emptyset$. For each $1$-type $\alpha$
choose $m+1$ minimal elements of type $\alpha$ in $\str{T}$ (or all of
them if there are less than $m+1$ such elements) and make them members
of $F_0$. Close $F_0$ under $\uparrow$, $\leftarrow$ and
$\rightarrow$, i.e., for each member of $F_0$ add to $F_0$ also all
its ancestors, siblings and all the siblings of its ancestors. This
finishes the construction of $F_0$.  Observe that $|F_0| \le (m+1)hd
|\AAA_\phi|$.

For each $v \in F_0$ and each conjunct of $\varphi$ of the form
$\phi_i =\forall x (\lambda_i(x) \rightarrow \exists y (x \fw y \wedge
\chi(x,y)))$ if $\str{T} \models \lambda_i[v]$ and there is no witness
for $v$ and $\phi_i$ in $F_0$ then find one in $\str{T}$ and add it to
$F_1$.  Similarly, For each $v \in F_1$ and each conjunct of $\varphi$
of the form $\phi_i =\forall x (\lambda_i(x) \rightarrow \exists y (x
\fw y \wedge \chi(x,y)))$ if $\str{T} \models \lambda_i[v]$ and there
is no witness for $v$ and $\phi_i$ in $F_0 \cup F_1$ then find one in
$\str{T}$ and add it to $F_2$.

Take as $F$ the smallest set containing $F_0 \cup F_1 \cup F_2$ and
closed under the relations $\uparrow$, $\leftarrow$ and
$\rightarrow$. Note that $|F_1| \le m|F_0| \le m(m{+}1)hd|\AAA_\phi|$,
and similarly 
$|F_2|\le m|F_1| \le m^2(m{+}1)hd|\AAA_\phi|$. It follows that  $|F| \le
(m{+}1)hd|\AAA_\phi| + \big(m(m{+}1)hd |\AAA_\phi|+ m^2(m{+}1)hd
|\AAA_\phi|\big)hd \le 3(m{+}1)^3h^2d^2 |\AAA_\phi|$, as required.

To obtain $\str{T}'$ we modify some $2$-types joining pairs of
elements in free position, one of which is in $T \setminus (F_0 \cup
F_1)$ and the other in $F_0$.  Consider any element $v \in T \setminus
(F_0 \cup F_1)$ and let $B$ be a minimal (with respect to $\subseteq$) 
set of elements providing the
required free witnesses for $v$ in $\str{T}$. Note that $|B| \le
m$. Let $\alpha$ be a $1$-type realized in $B$.  If all elements of
$1$-type $\alpha$ from $B$ are in $F_0$ then there is nothing to
do: we just retain the connections of $v$ with the elements of type $\alpha$ in $F_0$. 
Otherwise there are $m+1$ minimal realizations of $\alpha$ in
$F_0$, and at least $m$ of them is in free position to $v$.  Indeed,
$v$ cannot be an ancestor or a sibling of any of those $m+1$ minimal realizations of $\alpha$ (since $F_0$
is closed under $\uparrow$, $\leftarrow$ and
$\rightarrow$), so if it is not in free position to all  then
it is a descendant of one of them. But in this case it is in free position to all the other (since minimal realizations of $\alpha$ are in free position to each other).
 Thus, in this case, for any $w \in B$ of type $\alpha$ we can
choose a fresh $w'$ of type $\alpha$ in $F_0$ in free position to $v$
and set $\tp{T'}{v,w'}:= \tp{T}{v,w}$. We repeat this step
for all $1$-types of elements of $B$, thus ensuring that $v$ has all
the required free witnesses in $F_0$. We repeat this process for all
elements of $T \setminus (F_0 \cup F_1)$.

This finishes our construction of $\str{T}'$. Note that our surgery
does not affect the $2$-types inside $\str{T} \restr (F_0 \cup F_1)$
and the $2$-types joining the elements of $F_1$ with the elements of
$T \setminus (F_0 \cup F_1)$. Thus in $\str{T}'$ all elements of $F_0
\cup F_1$ retain their free witnesses in $F$ and all the remaining
elements have appropriate free witnesses in $F_0$ due to our
construction.  As we do not change the $2$-types joining the elements
which are not in free position thus all the upper, lower and sibling
witnesses are retained in $\str{T}'$.  Since $\str{T}'$ realizes only
$2$-types realized in $\str{T}$ the universal conjunct of $\forall xy
\chi(x,y)$ of $\varphi$ is satisfied in $\str{T}'$. Hence, $\str{T}'
{\models} \varphi$.
\end{proof}

\subsection{The algorithm}
We are now ready to present an alternating algorithm for the
finite satisfiability problem for \FOtall{}, working in exponential
time. Since \AExpTime=\ExpSpace{} this justifies
Thm.~\ref{t:fotbin}. Due to Lemma \ref{l:normalformbin} we can assume
that the input formula is given in normal form.

We first sketch our approach.  For a given normal form $\phi$ the
algorithm attempts to build a model $\str{T} \models \varphi$. It
first guesses its fragment $\str{F}$, of size exponentially bounded in
$|\varphi|$, intended to provide free witnesses for all elements of
$\str{T}$, and then expands it down.  Namely, it universally chooses
one of the leaves $v$ of $\str{F}$, guesses all its children $w_1,
\ldots, w_k$ (at most exponentially many), and guesses $2$-types
joining $w_i$-s with all their ancestors, with all elements of $\str{F}$,
and among each other. The algorithm verifies some consistency
conditions, and if succeeded then it universally chooses one of $w_i$
and proceeds with $w_i$ analogously like with $v$. This process is
continued until the algorithm decides that a leaf of $\str{T}$ is
reached.

We must ensure that the structure $\str{T}$ which is constructed by our algorithm is indeed a model of $\varphi$, i.e., all elements of $\str{T}$ have appropriate witnesses 
for $\forall \exists$ conjuncts, and that no pair of elements of $\str{T}$ violates the $\forall \forall$ conjunct. Note that when the algorithm inspects a node $v$ all its
siblings and ancestors are present in the memory. This allows to verify that $v$ has the required upper and sibling witnesses. Checking
the existence of  free witnesses is not problematic too, because, owing to Lemma \ref{l:fo2binglobalfree} we assume that they are provided by $\str{F}$, 
which is never removed from the memory. 
Verifying $\succv$-witnesses is also straightforward, since we guess all the children $w_1, \ldots, w_k$ of $v$ at once.
To deal with $\lessv$-witnesses the algorithm stores some additional data. Namely, together with each $w_i$ it guesses
the list of all $2$-types (called {\em promised $2$-types}) which will be assigned to the pairs consisting of $v$ or its ancestor and a descendant of $w_i$. 
This is obviously sufficient to see if $v$ will have the required $\lessv$-witnesses. The algorithm will take care of the consistency of the information about
promised types stored in various nodes, and then ensure that all the promised $2$-types will indeed be realized.

Turning to the problem of verifying that the universal conjunct of
$\varphi$ is not violated by any pair of elements of $\str{T}$ note
that it is easy for pairs of elements which are not in free position,
since at some point during the execution of the algorithm they are
both present in the memory and their $2$-type is then available. For a
pair of elements $u_1, u_2$ in free position there is an element $v$
such that $u_1$, $u_2$ are descendants of two different children of
$v$ from the list $w_1, \ldots, w_k$.  From information about the
promised $2$-types guessed together with $w_i$-s, we can extract the
list of $1$-types that will appear below each of $w_i$. Reading this
information we see that the $1$-types of $u_1$ and $u_2$ will appear
in free position, and we just need to verify that there is a $2$-type
consistent with the $\forall \forall$-conjunct which can join them.

Now we give a more detailed description of the algorithm.  It employs
a data structure, storing for each node $v$ the following components:
\begin{itemize}\itemsep0pt
\item $v$.{\tt 1{-}type} -- the $1$-type of $v$,
\item $v$.{\tt 2-type()} -- the function which for each $w$ being a sibling of $v$, an ancestor of $v$ or a member of $F$, returns the $2$-type of $(v,w)$,
\item $v$.{\tt promised-2-types()} -- a function which for each ancestor $w$ of $v$  returns a list of $2$-types, intended to contain all the $2$-types
which will be realized by $w$ with  descendants of $v$.
\end{itemize}
We assume that if a node $v$ is guessed then all the above components are constructed.

To avoid presentational clutter in the description of our algorithm we
omit some natural conditions on $2$-types guessed during its
execution, always assuming that they contain the intended navigational
atoms, i.e., the $2$-type joining an element with its child contains
$x \succv y$, with its right sibling $x \succh y$, and so on.

\medskip\noindent

\begin{algorithm}[H]
  \caption{\FOtall{}-sat-test}\label{algo:sat_test_f2}
  \begin{algorithmic}[1]

\REQUIRE a formula $\varphi$ in \FOtall{} normal form 

\STATE \textbf{guess} a tree $\str{F}$ of depth and degree of nodes bounded by $\mathfrak{f}(\varphi)$ and the number of nodes bounded by $3(m{+}1)^3 (\mathfrak{f}(\varphi))^4 |\AAA_\varphi|$
\FORALL {$v \in F$} \IF  {$v$ is not a leaf in $\str{F}$}
 \STATE \textbf{if }not \emph{consistent-with-ancestors-siblings-F}$(v)$ 
   \textbf{then reject \EndIf}
\STATE \textbf{if} not \emph{has-upper-sibling-free-witnesses}$(v)$ 
 \textbf{then reject \EndIf} 
\STATE  Let $w_1, \ldots, w_k$ be the list of the children of $v$
\STATE \textbf{if} not
  \emph{ensure-lower-witnesses}$(v,w_1,\ldots, w_k)$ \textbf{then reject \EndIf}
\STATE \textbf{if}  not \emph{propagates-promised-2-types}$(v,w_1,\ldots, w_k)$ \textbf{then reject \EndIf}
\STATE \textbf{if} not \emph{respects-universal-conjunct}$(v, w_1, \ldots, w_k)$
  \textbf{then reject \EndIf}
\ENDIF
\ENDFOR
\STATE \textbf{universally choose} a leaf $v$ of $\str{F}$; let $l$ be the depth of $v$ in $\str{F}$ 
 \WHILE {$l \le \mathfrak{f}(\varphi)$}
\STATE  \textbf{if} not \emph{consistent-with-ancestors-siblings-F}$(v)$ 
\textbf{then reject \EndIf} 
\STATE  \textbf{if} not \emph{has-upper-sibling-free-witnesses}$(v)$ \textbf{then reject \EndIf} 
\STATE \textbf{guess} a list $w_1, \ldots, w_k$ of children of $v$;  \textbf{if} $ k> \mathfrak{f}(\varphi)$
  \textbf{then reject \EndIf}
\STATE  \textbf{if} not \emph{ensure-lower-witnesses}$(v,w_1,\ldots, w_k)$
   \textbf{then reject \EndIf}
\STATE  \textbf{if}  not \emph{propagates-promised-2-types}$(v,w_1,\ldots, w_k)$
\textbf{then reject \EndIf}
\STATE \textbf{if} not \emph{respects-universal-conjunct}$(v, w_1, \ldots, w_k)$
   \textbf{then reject \EndIf}
\STATE \textbf{if} $k=0$ \textbf{then accept \EndIf} {\em \% $v$ is a leaf}
\STATE  \textbf{universally choose} $1 \le j \le k$ and set $v:=w_j$
    \ENDWHILE
\STATE \textbf{reject}
\end{algorithmic}
\end{algorithm}

The following function checks if all guessed components of $v$ are
consistent with the information about $v$'s siblings, ancestors and the
set $\str{F}$ of global free witnesses.

\floatname{algorithm}{Function}
\begin{algorithm}[H]
  \caption{\emph{consistent-with-ancestors-siblings-F}$(v)$}
  \begin{algorithmic}[1]

\FORALL{ $w$ being a sibling of $v$}
\STATE  let $\beta=v$.{\tt 2-type}$(w)$; \textbf{if} $w$.{\tt 2-type}$(v) \not=  \beta^{-1}$ 
  \textbf{then return  false \EndIf}
\ENDFOR
\IF {$v \in F$} 
\FORALL {$w \in F $}
\STATE let $\beta=v$.{\tt 2-type}$(w)$; \textbf{if} $w$.{\tt 2-type}$(v) \not=  \beta^{-1}$ 
\textbf{then return false \EndIf}
\ENDFOR
\ENDIF
\STATE \textbf{if} $v$ is the root  \textbf{then return true \EndIf}
\STATE let $u$ be the father of $v$
\FORALL { $w$ being an ancestor of $u$}
\STATE \textbf{if} $w.${\tt 2-type}$(v)$ $\not\in u.${\tt promised-2-types}$(w)$
\textbf{then return false \EndIf}
\ENDFOR
\RETURN true
  \end{algorithmic}
\end{algorithm}

The next function checks if $v$ has the required upper, sibling and
free witnesses.
\begin{algorithm}[H]
\caption{\emph{has-upper-sibling-free-witnesses(v)}}
\begin{algorithmic}
\FORALL {conjunct $\forall x (\lambda_i(x) \rightarrow \exists{y} (\theta_i(x,y) \wedge \chi_i(x,y)))$ of $\phi$\\ with $\theta_i \in \{ \tsuccv, \tlessv, \tsucch, \tprech, \tlessh, \tgreath, \tfree \}$}
\STATE\textbf{if} $v.${\tt 1-type} $\models \lambda_i(x)$ and there is no element $w$ being an ancestor or a sibling of $v$ or a memeber of $F$
  such that $v$.2-type$(w) \models \theta_i(x, y) \wedge \chi_i(x, y)$ \textbf{then return} false
\ENDFOR
\RETURN true
\end{algorithmic}
\end{algorithm}

The next function checks if the guess of $w_1, \ldots, w_k$ guarantees
lower witnesses for $v$.

\begin{algorithm}
  \begin{algorithmic}[1]
    
\caption{\emph{ensure-lower-witnesses}$(v,w_1,\ldots, w_k)$}
\FORALL{conjunct $\forall x (\lambda_i(x) \rightarrow \exists{y}\; \tsuccv (x,y) \wedge \chi_i(x,y))$ of $\phi$}
\IF {$v.${\tt 1-type} $\models \lambda_i(x)$ and there is no $w_i$ such that  $v$.{\tt 2-type}$(w_i) \models \chi_i(x, y)$} % \textbf{then return false \EndIf}
\RETURN false 
\ENDIF
\ENDFOR
\FORALL {conjunct $\forall x (\lambda_i(x) \rightarrow \exists{y} \;
  \tlessv (x, y) \wedge \chi_i(x,y))$ of $\phi$}
\STATE\textbf{if} $v.${\tt 1-type} $\models \lambda_i(x)$ and there is no $w_i$ such that for some $\beta \in w_i$.{\tt promised-2-types}$(v)$\\
\hshift $\beta \models \chi_i(x, y)$ \textbf{then return
  false \EndIf}
\ENDFOR
\RETURN true
  \end{algorithmic}
\end{algorithm}

The function below checks if the guess of $v$.{\tt promised-2-types()}
is propagated to the children of $v$ and consistent with $w_i$.{\tt
  promised-2-types()}.
\begin{algorithm}
  \begin{algorithmic}[1]
\caption{\emph{propagates-2-types}$(v, w_1, \ldots, w_k)$}
\FORALL{ $u$ being an ancestor of $v$}
 \STATE \textbf{if} {$v$.{\tt promised-2-types}$(u) \not= \bigcup_{i=1}^{k} \big( (\{ w_i$.{\tt 2-type}$(u))^{-1} \} \cup w_i$.{\tt promised-2-types}$(u)\big)$}\\ \textbf{then return false \EndIf}
\ENDFOR
\RETURN \textbf{true}
\end{algorithmic}
\end{algorithm}

The last function checks if the $2$-types formed by $v$ with all
elements of the constructed model (existing or promised) respect the
$\forall\forall$ conjunct.
\begin{algorithm}
  \begin{algorithmic}[1]
\caption{\emph{respects-universal-conjunct}$(v, w_1, \ldots, w_k)$}
\FORALL{ $u$ being an ancestor of $v$, a sibling of $v$, a member of $F$}
\STATE \textbf{if} $v$.{\tt 2-type}$(u) \not\models \chi(x, y)$ \textbf{then return false \EndIf}
\STATE\textbf{if} $(v.\texttt{2-type}(u))^{-1} \not\models \chi(x, y)$ \textbf{then return false \EndIf}
\ENDFOR
\STATE \textbf{if} $v$ is the root \textbf{then return true else} let $u$ be the father of $v$
\FORALL { $w_i$}  
\STATE let $desc_{w_i}:= \{\alpha: \exists \beta \in w_i.${\tt promised-2-types}$(u) \wedge \alpha = \beta \restr y \}$.\\
\% $desc_{w_i}$ \emph{is the list of promised $1$-types of descendants of $w_i$}
\ENDFOR
\FORALL {$i \not= j$}
\FORALL {$1$-type $\alpha$ from $desc_{w_i}$}
\FORALL {$1$-type $\alpha' \in desc_{w_j} \cup \{w_j.${\tt 1-type}$\}$}
\STATE \textbf{if}  there is no $2$-type $\beta$ such that $\beta \restr x=\alpha'$ and $\beta \restr y=\alpha$ and 
    $\beta(x, y) \models \tfree(x,y)\wedge\chi(x, y)$ \textbf{then return false}
\ENDFOR\ENDFOR\ENDFOR
\RETURN true
\end{algorithmic}
\end{algorithm}

\begin{lemma}
The procedure {\tt \FOtall{}-sat-test} works in alternating exponential time.
\end{lemma}
\begin{proof}
During its execution the algorithm guesses $\str{F}$, and builds a single path $P$ in $\str{T}$ together with the siblings of the elements from $P$.
The size of $\str{F}$ is bounded by $3(m+1)^3 (\mathfrak{f}(\varphi))^4|\AAA_\varphi|$, the length of $P$ and the degree of nodes are bounded by $\mathfrak{f}(\varphi)$, 
where $m$ is linear in $|\phi|$ and $\mathfrak{f}(\varphi)$ and $|\AAA_\varphi|$ are exponential in $|\phi|$. Thus the algorithm constructs exponentially
many nodes. For each node it guesses its $1$-type, $2$-types joining it with its siblings, ancestors and the elements of $\str{F}$ (exponentially
many in total) and promised $2$-types for each of its ancestors (again, information about the $2$-types for a single ancestor is bounded exponentially, since the total number of possible $2$-types is so bounded). The algorithm makes some consistency and correctness checking, which can be easily done in time polynomial in the size of the guesses. Hence the lemma follows.
\end{proof}

\begin{lemma}
The procedure \FOtall{}-\texttt{sat-test} accepts its input $\varphi$ iff $\varphi$ is satisfiable.
\end{lemma}
\begin{proof}(Sketch.)  Assume $\varphi$ has a model. By Theorem
  \ref{thm:fo2binsmallmodeltheorem} it has a model $\str{T}$ whose
  depth and degree of nodes are bounded by $\mathfrak{f}(\varphi)$.
  By Lemma \ref{l:fo2binglobalfree} there is a model $\str{T}'$ based
  on the same frame as $\str{T}$, in which one can distinguish a set
  $F$, of size at most $3(m+1)^3 (\mathfrak{f}(\varphi))^4
  |\AAA_\varphi|$, providing free witnesses for all elements of
  $\str{T}'$. Our algorithm can just take $\str{F}:=\str{T}' \restr F$
  and make all its guesses in accordance with $\str{T}$.

For the opposite direction assume that our algorithm has an accepting run. From this run we can naturally extract a partially defined tree structure $\str{T}$ and its substructure $\str{F}$. $\str{T}$ has defined
its tree frame, $1$-types of all nodes ($v$.{\tt 1-type} components), $2$-types of nodes not in free position and $2$-types of nodes in free position at least one of which is in $F$: the $2$-type joining $v$ and $w$ is stored in $v.${\tt 2-type}$(w)$ if $v$ is a descendant of $w$, or if $w \in F$ and $v \not\in F$, and in both $v.${\tt 2-type}$(w)$ and $w.${\tt 2-type}$(v)$ if $v$ and $w$ are siblings or $v,w \in F$. Note that the function 
\emph{consistent-with-ancestors-siblings-F} ensures that the $2$-types can be assigned without conflicts. This function, together with function \emph{propagates-2-types} ensures also the consistency of the information about promised $2$-types.

What is missing is $2$-types of pairs of elements $u_1, u_2$ in free
position none of which is in $F$. In this case there is an element $v$
such that $u_1$, $u_2$ are descendants of two different children of
$v$ from the list $w_1, \ldots, w_k$. Then, due to lines 7-10 of
the function \emph{respects-universal-conjunct}, there exists  a $2$-type
consistent with the $\forall \forall$-conjunct which can join them.

The constructed tree $\str{T}$ is indeed a model of $\varphi$:
\emph{respects-universal-conjunct} takes care of $\forall \forall$
constraint of $\varphi$, the sibling, upper and free witnesses are
ensured due to function \emph{has-upper-sibling-free-witnesses} and
lower witnesses are guaranteed by function
\emph{ensure-lower-witnesses} which uses the information about
promised-2-types.
\end{proof}

%%% Local Variables: 
%%% mode: latex
%%% TeX-master: "../main"
%%% End: 

%% file: C2.tex
\section{\texorpdfstring{$\CTwoFull$}{C2} on trees}
\label{sec:c2}
In this section we prove that the finite satisfiability problem for
$\CTwoFull$ over trees is \ExpSpace-complete. Intuitively, the proof
is a~combination of the two proofs from~\cite{WchKieroMazowieckiFO2Trees}
and \cite{WchPwitC2LinearOrder} that solve the problem for \FOt{} on
trees and for \Ct{} on linear orders respectively (note that a~linear
order is just a~tree whose each node has at most one child). However,
the method in~\cite{WchKieroMazowieckiFO2Trees} heavily depends on the
normal form from Definition~\ref{def:fo2normalform} where each
conjunct corresponds to at most one relative position $\theta \in
\Theta$. Although it is possible to bring a~$\CTwoFull$ formula into
an analogous normal form, it requires an exponential blowup (dividing
a~set of witnesses into 10 subsets corresponding to 10 order formulas
can be done in exponentially many ways). Therefore, to keep the
complexity under control, we stay with usual, less refined normal form
from Definition~\ref{def:c2normalform}, but to compensate it we
introduce a~novel technique combining type information with witness
counting.

\subsection{Multisets}

Any element of a model of a normal form conjunct $\forall
x\exists^{\bowtie C}y\ \chi $ may require up to $C$ witnesses, 
so we are interested in \textit{multisets} counting
these witnesses. To simulate counting up to the value $k$, we use the
function $\mathit{cut}_k : \mathbb{N} \rightarrow \lbrace 0, 1, 2,
\ldots, k, \infty \rbrace$, where $\mathit{cut}_k(i) = i$ for $i \leq
k$ and $\mathit{cut}_k(i) = \infty$ otherwise.

Formally, for a given $k \in \mathbb{N}$, a $k$-multiset $M$ of
elements from a set $S$ is a function $M : S \rightarrow \{ 0, 1, 2,
\ldots, k, \infty \}$. For every element $e \in S$ we simply define
$M(e)$, called the \textit{multiplicity} of $e$ in $M$, as the number
of occurrences $e$ in the multiset $M$, counted up to $k$. We employ
standard set-theoretic operations, i.e., union $\cup$ and
intersection $\cap$ with their natural semantics defined as follows:
for given multisets $A$ and $B$ and an arbitrary element $e$ from
their domains, we define $(A \cup B)(e) = \mathit{cut}_k \left( A(e) +
  B(e) \right)$ and $(A \cap B)(e) = \min(A(e), B(e))$. Additionally,
we define the empty multiset $\emptyset$ as the function that for any
argument returns $0$ and the singleton $\{ e \}$ of $e$ as the
function such that $\{ e \}(e) = 1$ and $\{ e \}(e') = 0$ for all $e'
\neq e$.

\subsection{Full types, witness counting and reduced types}

\begin{definition}[Full type]
  A $k$-\textit{full type} $\overline{\alpha}$ (over a signature
  $\tau=\tau_0 \cup \tau_{nav}$) is a function
  $\overline{\alpha} : \Theta \rightarrow \lbrace 0, 1, 2, \ldots, k,
  \infty \rbrace^{2^{\tau_0}}$, i.e., a function which takes a
  position from $\Theta$ and returns a $k$-multiset of $1$-types over
  $\tau$ that satisfies the following conditions:
\begin{itemize}
\item $\overline{\alpha}(\theta_{\uparrow}), \
  \overline{\alpha}(\theta_{\rightarrow}), \
  \overline{\alpha}(\theta_{\leftarrow})$ is either empty or a singleton,
\item $\overline{\alpha}(\theta_{=})$ is a~singleton, and
\item if $\overline{\alpha}(\theta_{\uparrow})$ 
    (respectively, $\overline{\alpha}(\theta_{\downarrow}), \
    \overline{\alpha}(\theta_{\rightarrow}), \
    \overline{\alpha}(\theta_{\leftarrow})$) is empty, then
  also the multiset $\overline{\alpha}(\theta_{\uparrow \uparrow^{+}}
  )$ (respectively, 
    $\overline{\alpha}(\theta_{\downarrow \downarrow^{+}}),
    \overline{\alpha}(\theta_{\rightrightarrows^{+}}),
    \overline{\alpha}(\theta_{\leftleftarrows^{+}})$)  is empty.
\end{itemize}

\end{definition}

Let $C$ be the function that for a given normal form $\phi$
returns $C(\phi)=\max \{C_i \}_{1 \le i \le m}$. We work with $k$-full
types usually in contexts in which a normal form $\phi$ is fixed, and
we are then particularly interested in $C(\phi)$-full types.
The purpose of a $k$-full type is to say for a given node $v$, for
each $\theta \in \Theta$ and each $1$-type $\alpha'$, how many
vertices (counting up to $k$) of $1$-type $\alpha'$ are in position
$\theta$ to $v$. Formally:
\begin{definition} For a given tree $\str{T}$ and $v \in T$ we denote
  by $\ftp{T}{k}{v}$ the unique $k$-full type \emph{realized} by $v$,
  i.e.,
  the $k$-full type $\overline{\alpha}$ such that
  $\overline{\alpha}(\theta_=)$ contains the $1$-type of $v$ and for
  all positions $\theta \in \Theta$ and for all atomic $1$-types
  $\alpha'$ we have that $$\overline{\alpha}(\theta)(\alpha') =
  \mathit{cut}_{k} \left( \# \lbrace w \in T \ : \ \str{T} \models
    \theta[v,w] \wedge \tp{T}{w}= \alpha' \rbrace \right).$$
\end{definition}

We next define functions which for a normal form $\phi$ and a
$C(\phi)$-full type $\overline{\alpha}$ say how many witnesses a
realization of $\overline{\alpha}$ has for each of the conjuncts of
$\phi$ in all possible positions $\theta$.

\begin{definition}[Witness counting
  functions] \label{Definition:WitnesscountingFunction} Let $\phi$ be
  a normal form formula, and let $\overline{\alpha}$ be a
  $C(\phi)$-full type. Assume that $\overline{\alpha}(\theta_=)$ =
  \{$\alpha$\}.  We associate with $\phi$ and $\overline{\alpha}$ a
  function $\wcf{\phi}{\overline{\alpha}}: \{1, \ldots, m \} \times
  \Theta \rightarrow \{0, 1, \ldots, C(\phi), \infty \}$, whose values
  are defined in the following way:
  \begin{itemize}
  \item for $\theta \in \lbrace \theta_=, \theta_\rightarrow,
    \theta_\leftarrow, \theta_\downarrow, \theta_\uparrow \rbrace$ and
    any $i$:
 $$\wcfp{\phi}{\overline{\alpha}}{i}{\theta} = \left\{\begin{array}{rl}
      1 & \mbox{if~}  \overline{\alpha}(\theta) {=} \lbrace \alpha' \rbrace  \text{~and~}  \alpha(x) {\wedge} \alpha'(y) {\wedge} \theta(x, y) {\models} \chi_i(x, y) \\
      0 & \text{otherwise},
\end{array}\right.$$

\item for $\theta \in \lbrace \theta_{\leftleftarrows^+},
  \theta_{\rightrightarrows^+}, \theta_{\downarrow \downarrow^+},
  \theta_{\uparrow \uparrow^+}, \theta_{\not\sim} \rbrace$ and any
  $i$:
		$$  \wcfp{\phi}{\overline{\alpha}}{i}{\theta} =\mathit{cut}_{C(\phi)} \left( \sum_{\alpha' : \ \alpha(x)  {\wedge} 
      \alpha'(y)  {\wedge}  \theta(x, y) \models \chi_i(x, y) }
    (\overline{\alpha}(\theta))(\alpha') \right ).$$
\end{itemize}
\end{definition}

This way $\wcfp{\phi}{\overline{\alpha}}{i}{\theta}$ is the number of
witnesses (counted up to $C(\phi)$), in relative position $\theta$,
for a node of full type $\overline{\alpha}$ and the formula $\chi_i$
from $\phi$.

Now we relate the notion of full types with the satisfaction of normal
form formulas.

\begin{definition}[$\varphi$-consistency] \label{def:phiconsistent}
  Let $\varphi$ be a $\CTwoFull$ formula in normal form. Let
  $\overline{\alpha}$ be a $C(\phi)$-full type.  Assume that
  $\overline{\alpha}(\theta_=)$ consists of a $1$-type $\alpha$. We
  say that $\overline{\alpha}$ is \textit{$\varphi$-consistent} if it
  satisfies the following conditions.

\begin{itemize}
\item $\alpha(x) \models \chi(x,x)$,
\item $\alpha(x) \wedge \alpha'(y) \wedge \theta(x,y) \models
  \chi(x,y) \big($ for every $\theta \in \Theta$, $\alpha' \in
  \overline{\alpha}(\theta) \; \big)$, and
\item for all $1 \le i \le m$ 
the inequality
  $\sum_{\theta \in \Theta} 
	\wcfp{\phi}{\overline{\alpha}}{i}{\theta}
  \bowtie_i C_i$ holds.
\end{itemize}

\end{definition}

\begin{lemma} \label{lemma:ficonsistent} Assume that a formula
  $\varphi \in \CTwoFull$ is in normal form. Then $\str{T} \models
  \varphi$ iff every $C(\phi)$-full type realized in $\str{T}$ is
  $\varphi$-consistent.
\end{lemma}

\begin{proof}

$\Longrightarrow$

Assume that $\str{T} \models \varphi$. Let $\overline{\alpha}$ be
a $C(\phi)$-full type realized in $\str{T}$. We have $\str{T} \models
\forall x \forall y \ \chi(x,y)$, and $\str{T} \models  \forall x \ \exists^{\bowtie_i C_i} y \ \chi_i(x,y) $
for all $i$. The first two conditions from
Definition~\ref{def:phiconsistent} are straightforward, because $\chi$
is true for every pair of vertices, every pair of vertices is related
by some $\theta \in \Theta$ and every vertex has it's own
$1$-type. For the third condition, take $v \in T$, such that
$\ftp{T}{C(\phi)}{v} = \overline{\alpha}$, and $i \in
\mathbb{N}$. The number of $w \in T$, such that $\str{T} \models
\chi_i[v,w]$, is $\bowtie_i C_i$, because $\str{T} \models  \forall x \ \exists^{\bowtie_i C_i} y \ \chi_i(x,y) $.

$\Longleftarrow$

Every pair of vertices is related with some $\theta \in \Theta$. Let
$\overline{\alpha}$, $\overline{\beta}$ be the $C(\phi)$-full types of nodes $v,
w \in T$ realized in $\str{T}$. By assumption
$\overline{\alpha}$, $\overline{\beta}$ are $\varphi$-consistent,
which proves (by the first and the second condition from Definition
\ref{def:phiconsistent}) that $\str{T} \models \forall x \forall
y \ \chi(x,y)$.

Fix a vertex $v \in T$, its $C(\phi)$-full type $\overline{\alpha} =
\ftp{T}{C(\phi)}{v}$ and some $i \in \mathbb{N}$. We know that
$\overline{\alpha}$ is $\varphi$-consistent, so $\sum_{\theta \in
  \Theta} \wcfp{\phi}{\overline{\alpha}}{i}{\theta} \bowtie_i C_i$. By
this fact, $\# \lbrace w \in T \ | \ \exists_{\theta \in \Theta} \
\theta(w,v) \rbrace \bowtie_i C_i$, which means that we have the right
number of witnesses for $v$ to satisfy the formula $\chi_i$. That
gives us $\str{T} \models  \forall x \ \exists^{\bowtie_i C_i} y \ \chi_i(x,y) $. We have shown that every
conjunct from $\varphi$ is true in $\str{T}$, so $\str{T}
\models \varphi$.
\end{proof}

The next notion will be used to describe information from full types 
in a (lossy) compressed form. We need this form to obtain tight complexity bounds.

\begin{definition}[$\phi$-reduced type] \label{def:reducedfulltype} Let
  $\varphi$ be a normal form $\CTwoFull$ formula.  For a given $C(\phi)$-full
  type $\overline{\alpha}$, its $\phi$-reduced form,
  $\rftpa{\varphi}{\overline{\alpha}}$,  is the tuple $\left( \alpha, \wcf{\phi}{\overline{\alpha}}, A, B, F \right)$, where $A =
  \overline{\alpha}(\theta_{\uparrow}) \cup
  \overline{\alpha}(\theta_{\uparrow \uparrow^{+}})$, $B =
  \overline{\alpha}(\theta_{\downarrow}) \cup
  \overline{\alpha}(\theta_{\downarrow \downarrow^{+}})$, $F =
  \overline{\alpha}(\theta_{\rightarrow}) \cup
  \overline{\alpha}(\theta_{\leftarrow}) \cup
  \overline{\alpha}(\theta_{\rightrightarrows^{+}}) \cup
  \overline{\alpha}(\theta_{\leftleftarrows^{+}}) \cup
  \overline{\alpha}(\theta_{\not\sim})$ and
  $\overline{\alpha}(\theta_=)$ is the singleton of the $1$-type
  $\alpha$. If the $C(\phi)$-full type $\overline{\alpha}$ is realized by
  a~vertex $v$ in $\str{T}$ then we say that 
	$\rftpa{\varphi}{\overline{\alpha}}$ is the $\phi$-reduced type of $v$.
	This reduced full type will be denoted also as $\rftpv{T}{\varphi}{v}$.
	\end{definition}

Intuitively, if a~$k$-full type $\overline{\alpha}$ is realized by a~vertex
$v$ in a~structure $\str{T}$ then the multisets $A,B,F$ in
$\rftpa{\varphi}{\overline{\alpha}}$ are respectively the
$k$-multisets of 1-types realized in $\str{T}$ above, below and in free
position to $v$.

Let $\overline{\alpha}, \overline{\beta}$ be $k$-full types. A
\textit{combined $k$-full type} is a $k$-full type $\overline{\gamma}$, such
that $\overline{\gamma}(\theta) = \overline{\alpha}(\theta)$ or
$\overline{\gamma}(\theta) = \overline{\beta}(\theta)$ for all
positions $\theta \in \Theta$.

\begin{lemma} \label{lemma:combinedfulltypes} Let $\overline{\alpha},
  \overline{\beta}$ be $\varphi$-consistent $C(\phi)$-full types such that their
  $\phi$-reduced forms are equal. Then the combined $C(\phi)$-full type
  $\overline{\gamma}$ in form $\overline{\gamma}(\theta) =
  \overline{\alpha}(\theta)$ for $\theta \in \lbrace
  \theta_{\uparrow}, \theta_{\uparrow \uparrow^+},$
  $\theta_{\rightarrow}, \theta_{\rightrightarrows^+},
  \theta_{\not\sim}, \theta_{\leftarrow}, \theta_{\leftleftarrows^+}
  \rbrace$ and $\overline{\gamma}(\theta) = \overline{\beta}(\theta)$
  for $\theta \in \lbrace \theta_=, \theta_{\downarrow},
  \theta_{\downarrow \downarrow^+} \rbrace$ is also
  $\varphi$-consistent.
\end{lemma}

\begin{proof}
  Obviously $\overline{\gamma}$ satisfies the first two conditions
  from Definition~\ref{def:phiconsistent} because $\overline{\alpha}$
  and $\overline{\beta}$ do. The third condition is guaranteed by the
  equality of the witness counting components. 
\end{proof}

\begin{example}
  Note that the assumption about equality of $\phi$-reduced full types, and
  in particular their witness counting components, is
  essential. In~\cite[Proposition 2]{WchKieroMazowieckiFO2Trees} the
  authors prove that in the setting without counting quantifiers
  a~combined type remains $\varphi$-consistent without the assumption
  about equality of the witness-counting components. The following example
  shows that in our scenario it is no longer true.

\begin{figure}[h] \label{f:exa}
\centering
\includegraphics[scale=0.52]{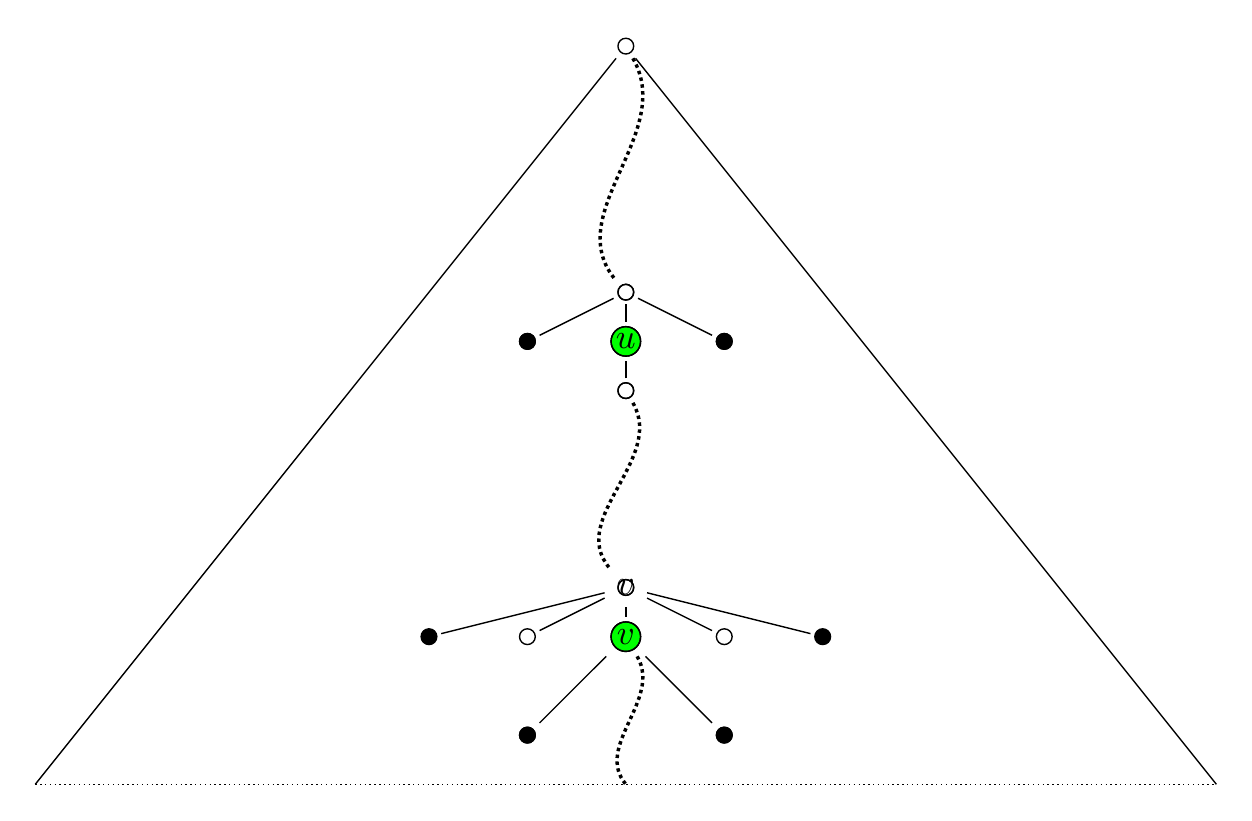}
\caption{Naive combination of full types}
\end{figure}

Let $\varphi$ be a~formula saying that  every green vertex has at most
three direct black neighbors below, on the left or on the right;
formally 
$$\varphi = \forall x \exists^{\leq 3} y \ \mathit{green}(x)
\Rightarrow \left( \mathit{black}(y) \wedge \left( x \downarrow y \vee
    y \downarrow x \vee x \rightarrow y \vee y \rightarrow x \right)
\right).$$
 Let $\str{T}$ be a tree model from Fig.~\ref{f:exa}. Denote $\overline{\alpha} = \ftp{T}{C(\phi)}{u}$ and
$\overline{\beta} = \ftp{T}{C(\phi)}{v}$. Because
$\str{T} \models \phi$, the $C(\phi)$-full types $\overline{\alpha}$ and $\overline{\beta}$ are
$\varphi$-consistent. However the combined $C(\phi)$-full type $\overline{\gamma}$, in form
described in Lemma \ref{lemma:combinedfulltypes}, is not
$\varphi$-consistent (the black nodes appear in $\overline{\gamma}$ on positions
$\theta_\downarrow$, $\theta_{\leftarrow}$, $\theta_{\rightarrow}$
four times in total).
\end{example}

\subsection{Small model theorem}
The general scheme of the decidability proof of finite satisfiability of
$\CTwoFull$ is similar to the one from Section~\ref{sec:f2bin}. Namely, we demonstrate
the small-model property of the logic, showing that every satisfiable
formula $\varphi$ has a tree model of depth and degree bounded
exponentially in $|\varphi|$. It is also obtained in a similar way, by
first shortening $\succv$-paths and then shortening the
$\succh$-paths. The technical details differ however. 

Recall that given a normal form $\phi$ we denote by $m$ the number of its
$\forall \exists$ conjuncts, and by $\AAA_\phi$ the set of $1$-types over
the signature consisting of the symbols appearing in $\phi$.

\begin{theorem}[Small model theorem]\label{thm:c2smallmodeltheorem}
  Let $\varphi$ be a formula of $\CTwoFull$ in normal form 
	If $\varphi$ is satisfiable then it has a a tree model in which every
    path has length bounded by $3 \cdot \left( C(\phi)+2 \right)^{10m + 1} \cdot |\AAA_\phi|^2$ and every vertex has degree bounded by
  $(4 C(\phi)^2 + 8C(\phi)) \cdot |\AAA_\phi|^5$.
\end{theorem}

  We split the proof of this theorem into two parts. In Section
  \ref{subsubsection:c2shortpaths} we show how to reduce the length of
  paths in a tree and in Section \ref{subsubsection:c2smalldegree} we
  show how to reduce the degree of every vertex.

\subsubsection{Short paths} \label{subsubsection:c2shortpaths}

\begin{lemma}
[Cutting lemma]
\label{thm:treecut}
  Let $\varphi \in \CTwoFull$ be a formula in normal form and
  $\str{T}$ be its tree model. If there are two vertices $u, v
  \in T$, such that $v$ is below $u$ and $\rftpv{T}{\varphi}{u}
  = \rftpv{T}{\varphi}{v}$, then the tree $\str{T}'$,
  obtained by replacing the subtree rooted at $u$ by the subtree
  rooted at $v$, is also a model of $\varphi$.
\end{lemma}

\begin{proof}
  The proof goes by case analysis. First, observe that the $C(\phi)$-full type
  of $u$ in tree $\str{T}'$ is a combination of the $C(\phi)$-full types of
  $u$ and $v$ in $\str{T}$ and thus, by Lemma
  \ref{lemma:combinedfulltypes}, it is $\varphi$-consistent.  In the
  rest of the proof we show that for every vertex $w$ in
  $\str{T}'$ we have $\ftp{{T}}{C(\phi)}{w} =
  \ftp{{T}'}{C(\phi)}{w}$. Then Lemma~\ref{lemma:ficonsistent}
  guarantees that the obtained tree $\str{T}'$ is indeed a model
  of $\varphi$.

  Let $w$ be any vertex from $T'$, or equivalently, a vertex from $T$
  such that $w$ lies inside the tree rooted at $v$ or lies outside of
  the tree rooted at $u$. There are four possible "locations" for $w$:

\begin{enumerate}[(a)]\itemsep0pt
\item $w$ is above $u$ (cf.~Fig.~\ref{f:casea})
\item $w$ is below $v$ (cf.~Fig.~\ref{f:caseb})
\item $w$ is in a free position to $u$  (cf.~Fig.~\ref{f:casec})
\item $w$ is a sibling of $u$
\end{enumerate}

For the rest of the proof, denote $\overline{\alpha} =
\ftp{\str{T}}{C(\phi)}{w}$ and $\overline{\beta} =
\ftp{\str{T}'}{C(\phi)}{w}$.  Let us consider the four cases distinguished above.

\begin{enumerate}[(a)]

\item $w$ is above $u$
\begin{figure}[!ht] 
  \centering
  \includegraphics[scale=0.6]{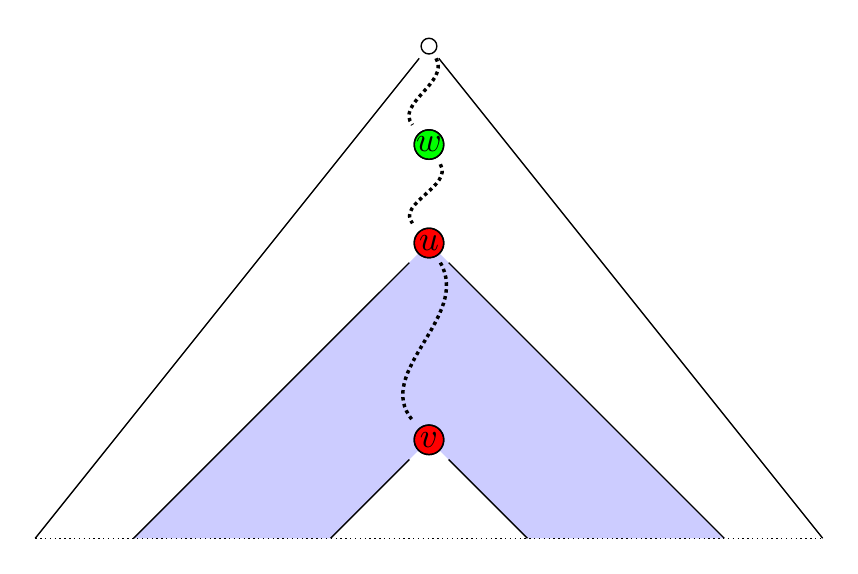}
	\caption{Case (a)} \label{f:casea}
\end{figure}

Of course cutting vertices between $u$ and $v$ does not change anything
above or in free position to $w$, so obviously for all $\theta \in \{
\teq, \tprecv, \tgreatv, \tprech, \tgreath, \tsucch, \tlessh,
\tfree\}$ we have $\overline{\alpha}(\theta) {=}
\overline{\beta}(\theta)$. The 1-type of the node immediately below
$w$ is also not changed, so $\overline{\alpha}(\tsuccv) =
\overline{\beta}(\tsuccv)$. The only missing case is equality of
multisets $\overline{\alpha}(\tlessv)$ and
$\overline{\beta}(\tlessv)$, but it follows from the equality of the
$\varphi$-reduced types $\rftpv{T}{\varphi}{u}$ and
$\rftpv{T}{\varphi}{v}$, and in particular their
$B$-components. More specifically, for a~given 1-type $\gamma$, the
number of occurrences of $\gamma$ in $\overline{\alpha}(\tlessv)$ is
(counted up to $C(\phi)$) the number of occurrences of $\gamma$ in the
subtree of $\str{T}$ rooted at $w$. We can divide this tree into
three pieces, as in Fig.~\ref{f:casea}: the upper part without the subtree
rooted at $u$, the lower part rooted at $v$ and the remaining middle
part. Now, the multiplicity of $\gamma$ in
$\overline{\alpha}(\tlessv)$ is simply the sum of
multiplicities  of $\gamma$ in each of these parts. But from the fact
that $\rftpv{T}{\varphi}{u}= \rftpv{T}{\varphi}{v}$, we
know that the multiplicity of $\gamma$ in the subtree rooted at $u$ is
the same as in the subtree rooted at $v$. It means that either there are
more than $C(\phi)$ occurrences of $\gamma$ in the subtree rooted at $v$ or
there are  no occurrences of $\gamma$ in the middle part. In both
cases the multiplicity of $\gamma$ below  $w$ is the same 
before and after the surgery.

\item $w$ is below $v$
\begin{figure}[!ht]
\centering
\includegraphics[scale=0.6]{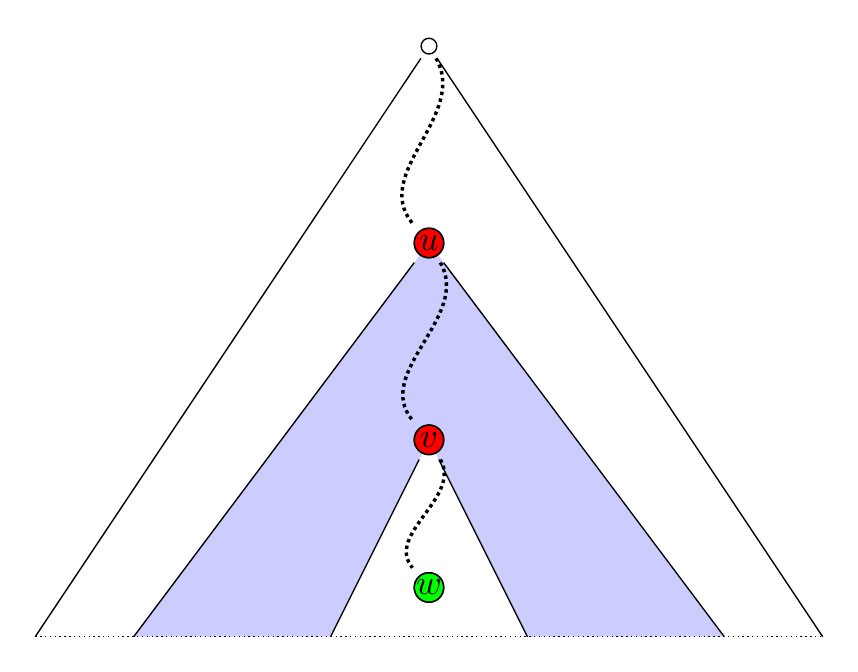}
\caption{Case (b)} \label{f:caseb}
\end{figure}

The reasoning to show that $\overline{\alpha}(\theta) =
\overline{\beta}(\theta)$ for $\theta \in \{ \teq, \tprecv, \tgreatv,
\tprech, \tgreath, \tsucch, \tlessh, \tlessv\}$ is the same or
symmetric to the previous case.

It is a~bit more tricky to prove that $\overline{\alpha}(\tfree) =
\overline{\beta} (\tfree)$. Consider an arbitrary 1-type
$\gamma$. Observe that vertices of type $\gamma$ in free position to $w$ in
$\str{T}$ are vertices of type $\gamma$ that are
\begin{itemize}
\item in free position to $w$ in the subtree rooted at $v$ 
(denote the number of them by $\# w_{\not\sim}$),
\item close and distant siblings of $v$ (denoted $\#
  v_{\rightleftarrows}$),
\item vertices in free position to $v$ in the subtree rooted at $u$
  (denoted $\# v_{\not\sim}$),
\item close and distant siblings of $u$ (denoted $\#
  u_{\rightleftarrows}$), or
\item in free position to $u$ in the whole tree ($\# u_{\not\sim}$).
\end{itemize}
This gives us the equation: 
\[ \left( \overline{\alpha}(\tfree)\right) \left( \gamma \right) = 
\mathit{cut}_{C(\phi)} \left( \# w_{\not\sim}
  {+} \# v_{\rightleftarrows} {+} \# v_{\not\sim} {+} \#
  u_{\rightleftarrows} {+} \# u_{\not\sim} \right).
\]

By the assumption that $\rftpv{T}{\varphi}{v} =
\rftpv{T}{\varphi}{u}$, the $F$-components of these $\varphi$-reduced $C(\varphi)$-full
types are equal, so $ \mathit{cut}_{C(\phi)} \left( \# v_{\rightleftarrows} {+}
  \# v_{\not\sim} {+} \# u_{\rightleftarrows} {+} \# u_{\not\sim} \right)
= \mathit{cut}_{C(\phi)} \left( \# v_{\rightleftarrows} {+} \# v_{\not\sim}
\right)$. It means that either $\mathit{cut}_{C(\phi)} \left( \#
  u_{\rightleftarrows} {+} \# u_{\not\sim} \right) = \infty$ or $\#
v_{\rightleftarrows} + \# v_{\not\sim} = 0$. In both cases, 
$\overline{\alpha}(\theta_{\not\sim})$ does not change  after removing
vertices between nodes $u$ and $v$.

\item $w$ is in a free position to $u$

\begin{figure}[!ht]
\centering
\includegraphics[scale=0.6]{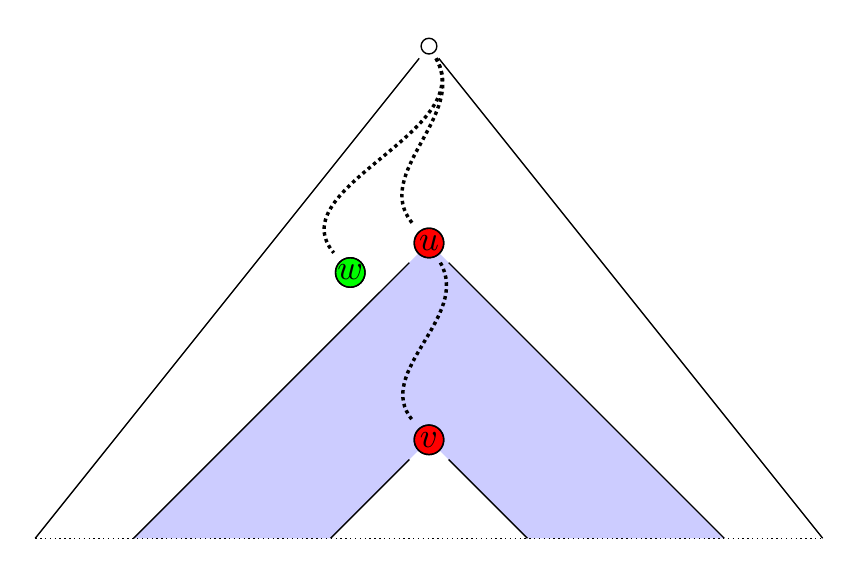}
\caption{Case (c)}  \label{f:casec}
\end{figure}

The fact that $\overline{\alpha}(\theta) = \overline{\beta}(\theta)$
for $\theta \in \lbrace \theta_{\downarrow}, \theta_{\downarrow
  \downarrow^+}, \theta_{\leftarrow}, \theta_{\leftleftarrows^+},
\theta_{\rightarrow}, \theta_{\rightrightarrows^+}, \theta_{\uparrow},
\theta_{\uparrow \uparrow^+}, \theta_= \rbrace$ is quite obvious (for
the same reason as in the previous cases).

We will show that $\overline{\alpha}(\theta_{\not\sim}) =
\overline{\beta}(\theta_{\not\sim})$. Observe that the whole tree
rooted at $u$ is in free position to $w$. $C(\phi)$-full types of all vertices
$w'$ outside this tree, such that $\str{T} \models
\theta_{\not\sim}[w,w']$ don't change, so we can concentrate only on
vertices from that tree. Note that $\rftpv{T}{\varphi}{u} =
\rftpv{T}{\varphi}{v}$, which means equality of "below"
multisets. Because multiplicity of each $1$-type $\gamma$ below $v$ is
equal to the multiplicity of $\gamma$ below $u$, we have equal
multiplicity of vertices in free position to $w$ before and after
replacing the root.

\item $w$ is a sibling of $u$

  The proof is similar to the previous ones. We only need to show that
  $\overline{\alpha}(\theta_{\not\sim}) =
  \overline{\beta}(\theta_{\not\sim})$. Observe that
  $\ftp{{T}}{C(\phi)}{u}(\theta_{\downarrow}) \ \cup \
  \ftp{{T}}{C(\phi)}{u}(\theta_{\downarrow \downarrow^+}) =
  \ftp{{T}}{C(\phi)}{v}(\theta_{\downarrow}) \ \cup \
  \ftp{{T}}{C(\phi)}{v}(\theta_{\downarrow \downarrow^+})$, so we
  don't accidentally cut any of the free witnesses of $w$, which
  proves the desired equality.
\end{enumerate}
\end{proof}

\begin{lemma} \label{lemma:c2smallpaths} Let $\varphi$ be a formula in
  normal form of $\CTwoFull$ satisfied in a finite tree. 
	Then there
  exists a tree model of $\varphi$ whose every $\downarrow$-path has
  length bounded by $3 \cdot \left( C(\phi)+2 \right)^{10m + 1} \cdot |\AAA_\phi|^2$.
\end{lemma}

\begin{proof}
  According to Lemma \ref{thm:treecut} we can restrict attention to models
  with the property that every $\phi$-reduced full type appears only once on
  every $\downarrow$-path. Let $\str{T} \models \varphi$ be a tree
  model with this property. Let $v_1, v_2, \ldots, v_n$ be a
  $\downarrow$-path in $\str{T}$. Observe that the $\phi$-reduced full types
  on this path behave in a~monotonic way in the sense that for every
  $i$ and the $\phi$-reduced full types of the $i, (i{+}1)$-th vertices $R_i =
  (\alpha_i, W_i, A_i, B_i, F_i)$ and $R_{i+1} =
  (\alpha_{i+1}, W_{i+1}, A_{i+1}, B_{i+1}, F_{i+1})$, we have
  $A_i {\subseteq} A_{i+1}, B_{i+1} {\subseteq} B_i$ and $F_i
  {\subseteq} F_{i+1}$. A $1$-type $\alpha$ can occur in a multiset
  from $0$ to $C(\phi)$ times. If $\alpha$ appears more than $C(\phi)$ times, its
  multiplicity is $\infty$. Hence the number of modifications of
  each multiset from $A,B,F$ is bounded by $ \left( C(\phi) + 2
  \right) \cdot |\AAA_\phi|$. There are up to $|\AAA_\phi|
  \cdot (C(\phi)+2)^{10m}$ $\phi$-reduced full types with fixed multisets $A,B,F$
  (because it is the number of all possible $1$-types multiplied by the
  number of all possible witness-counting functions). Combination of
  these two observations gives us the desired estimation $\left( C(\phi)+2 \right)^{10m + 1} \cdot |\AAA_\phi|^2 \cdot 3$.
\end{proof}

\subsubsection{Small degree} \label{subsubsection:c2smalldegree}
\begin{definition} \label{def:horizontaltype} For a given vertex $v
  \in T$ and its  $C(\phi)$-full type $\ftp{T}{C(\phi)}{v} =
  \overline{\alpha}$, the \emph{horizontal $C(\phi)$-full type} of $v$ in 
  $\str{T}$ is the quintuple 
\[\hftp{\str{T}}{C(\phi)}{v} =\left(
    \overline{\alpha}(\theta_=), \
    \overline{\alpha}(\theta_{\rightarrow}), \
    \overline{\alpha}(\theta_{\rightrightarrows^+}), \
    \overline{\alpha}(\theta_{\leftarrow}), \
    \overline{\alpha}(\theta_{\leftleftarrows^+})
  \right).
\]
\end{definition}

\begin{lemma} \label{lemma:c2smalldegfinitetree} Let $\varphi$ be a
  formula in normal form of $\CTwoFull$ satisfied in a finite tree
  $\str{T}$. 
	Then there exists a tree model of
  $\varphi$, obtained by removing some subtrees from $\str{T}$, such
  that the degree of every vertex is bounded by $(4 C(\phi)^2 + 8C(\phi)) \cdot |\AAA_\phi|^5$.
\end{lemma}

\begin{proof}
  First, we will show how to limit the degree of a single
  vertex. After that we can traverse the tree in depth-first manner
  and, by cutting unwanted vertices, obtain the desired model. Let $v$
  be a vertex from $T$. We denote by $\mathit{Children}(v)$ the set
  $\lbrace u_1, u_2, \ldots, u_N \rbrace$ of all children of $v$
  (ordered by $\rightarrow_+$). For every $1$-type $\alpha$, we are
  going to mark some elements of $\mathit{Children}(v)$ and, after
  that, limit number of vertices between two adjacent marked ones.

  Let $U_{\alpha} = \{ u_i \mid \tp{{T}}{u_i} = \alpha
  \}$ be the set of children of $v$ with the $1$-type $\alpha$ and let
  $U_{\alpha}^{\downarrow} = \{ u_i \mid \exists w {\in} T \left( u_i
    {\downarrow} w \vee u_i {\downarrow \downarrow^+} w \right) \wedge
  \tp{{T}}{w} = \alpha \}$ be the set of children of $v$
  with descendants of $1$-type $\alpha$.  For every $1$-type $\alpha$,
  we mark $\min \left( C(\phi), |U_{\alpha}| \right)$ vertices from
  $U_{\alpha}$ and $\min ( C(\phi), |U_{\alpha}^{\downarrow}| )$ vertices
  from $U_{\alpha}^{\downarrow}$ (it is important to mark $u_1$ and
  $u_N$ during this process). Marked vertices are the required witnesses for vertex
  $v$. The number of marked vertices ensure us that during the cutting
  we don't loose also any of free witnesses for any other vertex in the tree
  $\str{T}$. It's easy to see that during this process we marked
  at most $2C(\phi) \cdot |\AAA_\phi|$ vertices.

  Now the reasoning is similar to that of Lemma
  $\ref{thm:treecut}$. Let $u_i, u_j$ (where $i < j$) be two unmarked
  vertices, such that their horizontal $C(\varphi)$-full types are the same
  ($\hftp{\str{T}}{C(\phi)}{u_i} = \hftp{\str{T}}{C(\phi)}{u_j}$) and there
  are no marked vertices between them. Then the tree obtained by
  removing all vertices between $u_i$ and $u_j$, including $u_i$ and
  excluding $u_j$, together with subtrees rooted at them, is also a
  model of $\varphi$.  To prove it, first observe that cutting the
  vertices between $u_i$ and $u_j$ does not change any of vertices
  above and below $u_i$ and $u_j$. The marked vertices guarantee that
  none of vertices in the whole tree lost its free witnesses. Equality
  of the horizontal $C(\varphi)$-full types of $u_i$ and $u_j$ ensures that the
  numbers of right and left witnesses for vertices $u_i$ and $u_j$ are
  correct. Therefore the combined $C(\phi)$-full type of $u_j$, realized in the
  tree after the surgery, is $\varphi$-consistent. By Lemma
  \ref{lemma:ficonsistent}, the obtained tree is a model for the
  formula $\varphi$.

Continuing this process we can remove all vertices between the marked
pairs with the same horizontal $C(\varphi)$-full types. Observe that $\tlessh$ and $\tgreath$ components of the horizontal $C(\varphi)$-full types behave in the monotonic way. For  fixed $\tprech, \tsucch, \teq$ components of a given horizontal $C(\varphi)$-full type, the number of  its possible modifications on the path is bounded by $2 \cdot (C(\phi)+2) \cdot |\AAA_\phi|$.  This guarantees that between any adjacent marked vertices, we have at most $|\AAA_\phi|^3 \cdot 2 (C(\varphi) + 2) \cdot |\AAA_\phi|$  
  vertices. Using the
fact that we marked at most $2C(\phi) \cdot |\AAA_\phi|$
vertices and we know the upper bound on lengths of paths between
adjacent marked vertices, we can reduce the number of children of $v$
to $(4 C(\phi)^2 + 8C(\phi)) \cdot |\AAA_\phi|^5$.

By repeating this procedure as long as there are vertices of high
degree we obtain a desired model of $\varphi$.
\end{proof}

\subsection{Algorithm}

In this section we design an algorithm checking if a given formula
$\varphi \in \CTwoFull$ has a finite tree model. First, by Lemma
\ref{lemma:normalform}, we can assume that $\varphi$ is in normal
form. Second, by Theorem \ref{thm:c2smallmodeltheorem}, we can restrict attention
to models with exponentially bounded vertex degree and
$\downarrow$-path length.

We will present an alternating algorithm working in exponential space.
The idea of the algorithm is quite simple. For each vertex $v$ we will
guess its $C(\phi)$-full type and check if it is $\varphi$-consistent. If it is,
we guess the $v$'s children and their full types. After that, we check
if their $C(\phi)$-full types are locally consistent (see the procedure below),
which guarantees that we guessed correctly. The algorithm starts with
$v = $ root and works recursively with its children. The procedure
presented here is a modification of the one from
\cite{WchKieroMazowieckiFO2Trees}.

\floatname{algorithm}{Procedure}
\begin{algorithm}[H]
  \caption{Checking if given $C(\varphi)$-full types are
    locally-consistent}\label{algo:localycomp}
  \begin{algorithmic}[1]
    \REQUIRE $C(\varphi)$-Full types $\overline{\alpha}$, $\overline{\alpha_1}$, $\ldots$, $\overline{\alpha_k}$   \\
    \STATE \textbf{Return} \textbf{True} if all of the statements below are true. \textbf{Return} \textbf{False} otherwise.
	\STATE $\overline{\alpha_i}(\theta_{\leftarrow}) = \overline{\alpha_{i-1}}(\theta_{=} )$ for $i > 1$ and $\overline{\alpha_1}(\theta_{\leftarrow}) = \emptyset$
	\STATE $\overline{\alpha_i}(\theta_{\leftleftarrows^+}) = \overline{\alpha_{i-1}}(\theta_{\leftarrow} ) \cup \overline{\alpha_{i-1}}(\theta_{\leftleftarrows^+} )$ for $i > 1$ and $\overline{\alpha_1}(\theta_{\leftleftarrows^+}) = \emptyset$
	\STATE $\overline{\alpha_i}(\theta_{\rightarrow}) = \overline{\alpha_{i+1}}(\theta_{=} )$ for $i < k$ and $\overline{\alpha_k}(\theta_{\rightarrow}) = \emptyset$
	\STATE $\overline{\alpha_i}(\theta_{\rightrightarrows^+}) = \overline{\alpha_{i+1}}(\theta_{\rightarrow} ) \cup \overline{\alpha_{i+1}}(\theta_{\rightrightarrows^+} )$ for $i < k$ and $\overline{\alpha_k}(\theta_{\rightrightarrows^+}) = \emptyset$
	\STATE $\overline{\alpha}(\theta_{\downarrow}) = \bigcup_{j=1}^{k} \overline{\alpha_j}(\theta_=)$
	\STATE $\overline{\alpha}(\theta_{\downarrow \downarrow^+}) = \bigcup_{i=1}^{k} \left( \overline{\alpha_i}(\theta_{\downarrow}) \cup \overline{\alpha_i}(\theta_{\downarrow \downarrow^+ }) \right)$
	\STATE for $1 \leq i \leq k:$ $\overline{\alpha_i}(\theta_{\uparrow}) = \overline{\alpha}(\theta_=)$
	\STATE for $1 \leq i \leq k: $\\ $\overline{\alpha_i}(\theta_{\not\sim}) = \overline{\alpha}(\theta_{\not\sim}) \cup \overline{\alpha}(\theta_{\leftarrow}) \cup \overline{\alpha}(\theta_{\rightarrow}) \cup \overline{\alpha}(\theta_{\leftleftarrows^+}) \cup \overline{\alpha}(\theta_{\rightrightarrows^+}) \cup \bigcup_{j \neq i} \left( \overline{\alpha_j}(\theta_{\downarrow}) \cup \overline{\alpha_j}(\theta_{\downarrow \downarrow^+}) \right)$
  \end{algorithmic}
\end{algorithm}

\begin{algorithm}
  \caption{Satisfiability test for $\CTwoFull$}\label{algo:sat_test_c2}
  \begin{algorithmic}[1]
    \REQUIRE Formula $\varphi \in \CTwoFull$ in normal form.\\
    \STATE Let MaxDepth := $3 \cdot \left( C(\phi)+2 \right)^{10m + 1} \cdot |\AAA_\phi|^2$
    \STATE Let MaxDeg := $(4 C(\phi)^2 + 8C(\phi)) \cdot |\AAA_\phi|^5$
    \STATE Lvl := $0$.
    \STATE \textbf{guess} a $C(\varphi)$-full type $\overline{\alpha}$  s.t. $\overline{\alpha}(\theta)=\emptyset$ for $\theta \in \lbrace \theta_{\uparrow},\theta_{\uparrow \uparrow^+},\theta_{\rightarrow},\theta_{\leftarrow},\theta_{\rightrightarrows^+}, \theta_{\leftleftarrows^+},\theta_{\not\sim} \rbrace$.
    \WHILE{Lvl $<$ MaxDepth }
    	\STATE \textbf{if} $\overline{\alpha}$ is not $\varphi$-consistent \textbf{then} \textbf{reject}
    	\STATE \textbf{if} $\overline{\alpha}(\theta_{\downarrow}) = \overline{\alpha}(\theta_{\downarrow \downarrow^+}) = \emptyset$ \textbf{then} \textbf{accept}
    	\STATE \textbf{guess} an integer $1 \leq k \leq$ MaxDeg
    	\STATE \textbf{guess} $C(\phi)$-full types $\overline{\alpha_1}$, $\overline{\alpha_2}$, $\ldots$, $\overline{\alpha_k}$
    	\STATE \textbf{if not} locally-consistent$\left(\overline{\alpha}, \overline{\alpha_1}, \overline{\alpha_2}, \ldots, \overline{\alpha_k} \right)$ \textbf{then reject}
    	\STATE Lvl $:=$ Lvl $+$ $1$
    	\STATE \textbf{universally choose} $1 \leq i \leq k$; let $\overline{\alpha} = \overline{\alpha_i}$
    \ENDWHILE
    \STATE \textbf{reject}
  \end{algorithmic}
\end{algorithm}

\begin{lemma}
  Procedure \ref{algo:sat_test_c2} accepts its input $\varphi$ iff
  $\varphi$ is satisfiable.
\end{lemma}

\begin{proof}
  Assume  $\varphi$ is satisfiable. Then there
  exists a small tree model $\str{T}$ as guaranteed by
  Theorem \ref{thm:c2smallmodeltheorem}. We can run the algorithm and
  guess exactly the same $C(\phi)$-full types as in $\str{T}$. The guessed
  $C(\varphi)$-full types are locally-consistent and $\varphi$-consistent, so
  procedure \ref{algo:sat_test_c2} accepts.

  Assume that Procedure \ref{algo:sat_test_c2} accepts its input
  $\varphi$. Then we can reconstruct the tree $\str{T}$ from the
  received $C(\phi)$-full types. The guessed $C(\varphi)$-full types are
  $\varphi$-consistent, which guarantees that we have the right number
  of witnesses to satisfy the formula. Moreover, the function
  locally-consistent ensures that the $C(\phi)$-full types realized in
  $\str{T}$ are indeed as we guessed. By Lemma
  \ref{lemma:ficonsistent}, $\str{T}$ is a tree model for
  $\varphi$ and thus $\varphi$ is satisfiable.

\end{proof}

\begin{theorem}
  The satisfiability problem for $\CTwoFull$ over finite trees is
  \ExpSpace-complete.
\end{theorem}

\begin{proof}
 Our procedure works in alternating exponential time, since maximum degree
  and path length are bounded exponentially in $| \varphi |$. The \ExpSpace{}-upper
	bound follows from the well know fact that
  \AExpTime=\ExpSpace{}. The \ExpSpace-lower bound comes from
  \cite{BBCKLMW-tocl}.
\end{proof}

%%% Local Variables: 
%%% mode: latex
%%% TeX-master: "../main"
%%% End: 

%% file: expressivepower.tex
\section{Expressive power}

A natural question is whether adding counting quantifiers increases
the expressive power of two-variable logic over trees.  We answer this
question concentrating on the classical scenario assuming that
signatures contain no common binary symbols.  Under this scenario
\FOt$[\succv, \lessv, \succh, \lessh]$ is known to be expressively
equivalent to the navigational core of XPath
\cite{MarxDeRijke04}. Here we show that \Ct{}$[\succv, \lessv, \succh,
\lessh]$ shares the same expressivity.

Let us note, however, that it is the presence of the sibling relations
which makes \FOt{} and \Ct{} equivalent. Indeed, over unordered trees
\FOt{} cannot count:

\begin{theorem}
\FOt$[\succv, \lessv]$ is less expressive than \Ct{}$[\succv, \lessv]$.
\end{theorem}
\begin{proof}
Let us assume that the signature contains no unary predicates and for $i \in \N$ let $\str{T}_i$ denote the tree consisting just of a root and its $i$ children. 
Obviously $\str{T}_3 \models \exists x \exists^{\ge 3} y \; x \lessv y$ while $\str{T}_2 \not\models \exists x \exists^{\ge 3} y \; x \lessv y$. 
On the other hand, $\str{T}_2$ and $\str{T}_3$ are indistinguishable in \FOt$[\succv, \lessv]$. It can be seen by observing that Duplicator
has a simple winning strategy in the standard two-pebble game of any length
played on $\str{T}_2$ and $\str{T}_3$.
\end{proof}

Now we turn to the advertised equivalence of \FOt{} and \Ct{} in the case of full navigational signature.

\begin{theorem} \label{t:equivalence}
\FOt$[\succv, \lessv, \succh, \lessh]$ and \Ct{}$[\succv, \lessv, \succh, \lessh]$ are expressively equivalent.
\end{theorem}

We give a detailed proof. First we show that one can say in \FOt$[\succv, \lessv, \succh, \lessh]$ that for a given
node $v$ there are at least $k$ nodes in some specific \emph{position} to $v$ that have a fixed unary property $\psi$ expressible in \FOt{}. 
Let us define the set $\cal P$ of
positions we are interested in. Some of the positions correspond directly to the order formulas from $\Theta$, but for technical reasons we need to introduce also some other.
We represent the positions with help of graphical symbols. Intuitively, the crossed circle corresponds to $v$, the filled circles correspond to nodes among
which we look for those satisfying $\psi$ and the empty circles are auxiliary.
We distinguish sixteen positions:

${\cal P} = \lbrace \posone, \postwo, \posthree, \posfour, \posfourteen, \posfifteen, %$ $\hspace*{100pt} %tak bylo w~dwukolumnowym
\posfive, \possix, \posthirteen, \posseven, \poseight, \posnine\\$
$\hspace*{250pt} \posten, \poseleven, \possixteen, \postwelve \rbrace.$

Let us formalize the given intuitive meaning of the introduced
symbols.  Let $\str{T}$ be a tree, $v$ its node, $c$ a natural number,
$pos \in {\cal P}$, and $\psi$ any \FOt$[\succv, \lessv, \succh,
\lessh]$ formula with one free variable, We say that $v$
\emph{satisfies} property $\Wprop{c}{pos}{\psi}$ if there are at least
$c$ nodes $w$ such that $\str{T} \models \psi[w]$ and $w$ is in
position $pos$ to $v$, i.e.,
\begin{itemize}\itemsep0pt
\item if $pos=\posone$ then $w$ is a descendant of $v$,
\item if $pos=\postwo$ then $w=v$ or $w$ is a descendant of $v$,
\item if $pos=\posthree$ then $w$ is a following-sibling of $v$ or a descendant of a following-sibling of $v$, 
\item if $pos=\posfour$ then $w$ is a descendant of a following-sibling of $v$,
\item if $pos=\posfourteen$ then $w$ is a preceding-sibling of $v$ or a descendant of a preceding-sibling of $v$,
\item if $pos=\posfifteen$ then $w$ is a descendant of a preceding-sibling of $v$,
\item if $pos=\posfive$ then $w$ is a child of $v$,
\item if $pos=\possix$ then $w$ is a descendant of $v$ but not its child,
\item if $pos=\posthirteen$ then $w$ is an ancestor of  $v$,
\item if $pos=\posseven$ then $w$ is an ancestor of $v$ but not its father,
\item if $pos=\poseight$ then $w$ is a following-sibling of $v$,
\item if $pos=\posnine$ then $w$ is a preceding-sibling of $v$,
\item if $pos=\posten$ then $w$ is a following-sibling of $v$ but not the closest one,
\item if $pos=\poseleven$ then $w$ is a preceding-sibling of $v$ but not the closest one,
\item if $pos=\possixteen$ then $w$ is a sibling of $v$ or a descendant of a sibling of $v$,
\item if $pos=\postwelve$ then $w$ a sibling of an ancestor of $v$ or a descendant of a sibling of an ancestor of $v$.
\end{itemize}

\begin{lemma}\label{l:definitions}
For any $c \in \N$,  any \FOt$[\succv, \lessv, \succh, \lessh]$ formula $\psi$ with one free variable,
and $pos \in {\cal P}$,
there is an \FOt$[\succv, \lessv, \succh, \lessh]$ formula $\Psiprop{c}{pos}{\psi}$ with one free variable, such that for any tree $\str{T}$ and $v \in T$ 
we have $\str{T} \models \Psiprop{c}{pos}{\psi}[v]$ iff $v$ satisfies $\Wprop{c}{pos}{\psi}$.
\end{lemma}
\begin{proof}
The proof goes by induction on $c$. The base case $c=1$ is  straightforward:
\begin{itemize}\itemsep0pt
\item
$\Psiprop{1}{\posone}{\psi}(x) \equiv \exists y \; (x \lessv y \wedge \psi(y))$
\item 
$\Psiprop{1}{\postwo}{\psi}(x) \equiv \psi(x) \vee \exists y \; (x \lessv y \wedge \psi(y))$
\item
$\Psiprop{1}{\posthree}{\psi}(x) \equiv \exists y \; (x \lessh y \wedge \psi(y)) \vee \exists y \; (x \lessh y \wedge \exists x (y \lessv x \wedge \psi(x)))$
\item
$\Psiprop{1}{\posfour}{\psi}(x) \equiv \exists y \; (x \lessh y \wedge \exists x (y \lessv x \wedge \psi(x)))$
\item 
$\Psiprop{1}{\posfourteen}{\psi}(x)$ analogously
\item 
$\Psiprop{1}{\posfifteen}{\psi}(x)$ analogously
\item
$\Psiprop{1}{\posfive}{\psi}(x) \equiv \exists y \; (x \succv y \wedge \psi(y))$
\item
$\Psiprop{1}{\possix}{\psi}(x) \equiv \exists y \; (x \lessv y \wedge \neg (x \succv y) \wedge \psi(y))$
\item 
$\Psiprop{1}{\posthirteen}{\psi}(x) \equiv \exists y \; (y \lessv x \wedge \psi(y))$
\item
$\Psiprop{1}{\posseven}{\psi}(x) \equiv \exists y \; (y \lessv x \wedge \neg ( y \succv x) \wedge \psi(y))$
\item
$\Psiprop{1}{\poseight}{\psi}(x) \equiv \exists y \; (x \lessh y \wedge \psi(y))$
\item
$\Psiprop{1}{\posnine}{\psi}(x)$ analogously
\item
$\Psiprop{1}{\posten}{\psi}(x) \equiv \exists y \; (x \lessh y \wedge \neg (x \succv y) \wedge \psi(y))$
\item
$\Psiprop{1}{\poseleven}{\psi}(x)$ analogously 
\item
$\Psiprop{1}{\possixteen}{\psi}(x) \equiv \exists y \; ((x \lessh y \vee (y \lessh x) \wedge (\psi(y) \vee \exists x (y \lessv x \wedge \psi(x)))$
\item
$\Psiprop{1}{\postwelve}{\psi}(x) \equiv \exists y \; (y \lessv x \wedge \exists x ((y \lessh x \vee x \lessh y) \wedge (\psi(x) \vee \exists y (x \lessv y \wedge \psi(y))))) $
\end{itemize}
Assume now that the desired $\Psiprop{c}{pos}{\psi}$ formulas exist for all $1 \le c < k$. We show how
to define $\Psiprop{k}{pos}{\psi}$ using $\Psiprop{c}{pos'}{\psi}$ for $c<k$, or $c=k$ but in this case for $pos'$ defined in one of the earlier items.
If in any definition $\Psiprop{c}{pos}{\psi}$ with $c=0$ appears  it is replaced by $\top$.
\begin{itemize}\itemsep0pt
\item
$\Psiprop{k}{\posone}{\psi}(x) \equiv$\\ 
$\exists y (x \lessv y \wedge ((\psi(y) \wedge \Psiprop{k-1}{\posone}{\psi}(y)) \vee \bigvee_{i \in[1,k-1]} (\Psiprop{i}{\postwo}{\psi}(y) \wedge \Psiprop{k-i}{\posthree}{\psi}(y)) )$
\item 
$\Psiprop{k}{\postwo}{\psi}(x) \equiv$ $\psi(x) \wedge \Psiprop{k-1}{\posone}{\psi}(x) \vee \Psiprop{k}{\posone}{\psi}(x)$ 
\item
$\Psiprop{k}{\posthree}{\psi}(x) \equiv$ $\exists y (x \lessh y \wedge \bigvee_{i \in [1,k]}(\Psiprop{i}{\postwo}{\psi}(y) \wedge \Psiprop{k-i}{\posthree}{\psi}(y)  ))$
\item
$\Psiprop{k}{\posfour}{\psi}(x) \equiv$ $\exists y (x \lessh y \wedge \bigvee_{i \in [1,k]}(\Psiprop{i}{\posone}{\psi}(y) \wedge \Psiprop{k-i}{\posfour}{\psi}(y)  ))$
\item
$\Psiprop{k}{\posfourteen}{\psi}(x)$ analogously
\item
$\Psiprop{k}{\posfifteen}{\psi}(x)$ analogously
\item
$\Psiprop{k}{\posfive}{\psi}(x) \equiv$ $\exists y (x \succv y \wedge \psi(y) \wedge \Psiprop{k-1}{\poseight}{\psi}(y))$
\item
$\Psiprop{k}{\possix}{\psi}(x) \equiv$ $\exists y (x \succv y \wedge \bigvee_{i \in [1,k]}(\Psiprop{i}{\posone}{\psi}(y) \wedge \Psiprop{k-i}{\posfour}{\psi}(y)  ))$
\item
$\Psiprop{k}{\posthirteen}{\psi}(x) \equiv$ $\exists y (y \lessv x \wedge \psi(y) \wedge \Psiprop{k-1}{\posthirteen}{\psi}(y))$
\item
$\Psiprop{k}{\posseven}{\psi}(x) \equiv$ $\exists y (y \lessv x \wedge \neg(y \succv x) \wedge \psi(y) \wedge \Psiprop{k-1}{\posthirteen}{\psi}(y))$
\item
$\Psiprop{k}{\poseight}{\psi}(x) \equiv $ $\exists y (x \lessh y \wedge \psi(y) \wedge \Psiprop{k-1}{\poseight}{\psi}(y))$
\item
$\Psiprop{k}{\posnine}{\psi}(x) \equiv$  analogously
\item
$\Psiprop{k}{\posten}{\psi}(x) \equiv$  $\exists y (x \lessh y \wedge \neg(x \succh y) \wedge \psi(y) \wedge \Psiprop{k-1}{\poseight}{\psi}(y))$
\item
$\Psiprop{k}{\poseleven}{\psi} \equiv $  analogously
\item
$\Psiprop{k}{\possixteen}{\psi} \equiv $ $\bigvee_{i \in [0,k]} (\Psiprop{i}{\posthree}{\psi}(x) \wedge \Psiprop{k-i}{\posfourteen}{\psi}(x))$
\item
$\Psiprop{k}{\postwelve}{\psi} \equiv $ $\exists y (y \lessv x \wedge \bigvee_{i \in [1,k]}( \Psiprop{i}{\possixteen}{\psi}(y) \wedge \Psiprop{k-i}{\postwelve}{\psi}(y))).$
\end{itemize}
Most of the above equivalences are obvious. As an example, let us explain the first one. Assume that $\str{T} \models \Psiprop{k}{\posone}{\psi}[v]$. Choose
$k$ descendants of $v$ satisfying $\psi$. Let $u$ be the maximal element of $\str{T}$ such that all the chosen elements are in the subtree of $u$.
If $u$ is one of the chosen elements then $\str{T} \models v \lessh u \wedge \psi[u] \wedge \Psiprop{k-1}{\posone}{\psi}[u]$. Otherwise take the leftmost
child $w$ of $u$ such that the subtree of $w$ contains at least one of the chosen elements. Note that the subtree of $w$ contains at most 
$k-1$ chosen elements since otherwise it would contradict the maximality of $u$. Thus
$\str{T} \models v \lessv w \wedge \bigvee_{i \in[1,k-1]} (\Psiprop{i}{\postwo}{\psi}[w] \wedge \Psiprop{k-i}{\posthree}{\psi}[w]) )$.
The opposite direction is obvious. 
\end{proof}

We are now ready to show Theorem \ref{t:equivalence}. It follows from the following lemma. 
\begin{lemma}
Let $\phi$ be a \Ct{}$[\succv, \lessv, \succh, \lessh]$ formula with at most one free variable. There exists an \FOt{}$[\succv, \lessv, \succh, \lessh]$
formula $trans(\phi)$ such that for any tree $\str{T}$, and any $v \in T$ we have $\str{T} \models \phi[v] \iff \str{T} \models trans(\phi)[v]$.
\end{lemma}

\begin{proof}
In our translation process we will work with formulas using both counting quantifiers and standard existential quantifiers. 
Due to the equivalence $\exists^{\le c} x \psi \equiv \neg \exists^{\ge c+1} x \psi$ we can assume that 
all counting quantifiers of the form $\exists^{\ge c}$. 
We  take a most deeply nested subformula of $\phi$ of the form 
$\exists^{\ge k} y \psi(x,y)$. Thus $\psi(x,y)$ is a boolean combination of atoms and  \FOt{}$[\succv, \lessv, \succh, \lessh]$ formulas
starting with the existential quantifier.

Let us convert $\psi(x,y)$ into disjunctive normal form, $\psi(x,y) \equiv \psi_1(x,y) \vee \ldots, \vee \psi_l(x,y)$, such that $\psi_i(x,y)$ and $\psi_j(x,y)$
are mutually exclusive for $i \not=j$.
Let ${\cal F}$ be the set of functions $f$ of type $[0,k]^{[1, l]}$ such that $\sum_{i=1}^{l} f(i)=k$. 
Intuitively, such a function specifies how many of $k$ witnesses for $\exists^{\ge k} \psi$ are witnesses for $\psi_i$. 
We can now write 
$\exists^{\ge k} y \psi(x,y)$ equivalently as $\bigvee_{f \in {\cal F}} \bigwedge_{i=1}^l \exists^{\ge f(i)} y \psi_i(x,y)$.
Here and later we assume that if a subformula starting with $\exists^{\ge 0}y$ appears in our process then it is
immediately replaced by $\top$.
Our task reduces now to translating $\exists^{\ge k} y \psi_i(x,y)$ for $\psi_i$ being a conjunction of atoms or $\FOt$ subformulas
starting with $\exists$.

Further, let us replace $\psi_i(x,y)$ by $\bigvee_{\theta \in \Theta} (\theta(x,y) \wedge \psi_i(x,y))$.
Consider the set ${\cal G}$ of functions $g$ of type $[0,k]^{\Theta}$, such that $\sum_{\theta \in \Theta} g(\theta) = k$ and
$g(\theta) \in \{0, 1 \}$ for $\theta \in \{  \tprecv, \tsucch, \tprech, \teq \}$.
Observe that  $\exists^{\ge k} y \psi_i(x,y)$ is equivalent to $\bigvee_{g \in {\cal G}} \bigwedge_{\theta \in \Theta}
\exists^{\ge g(\theta)} y (\theta(x,y) \wedge \psi_i(x,y))$.

It remains to take care of formulas of the form $\exists^{\ge k} y (\theta(x,y) \wedge \psi_i(x,y))$.
Let $\psi'_i(x,y)$ be the result of replacing in $\psi_i(x,y)$ every binary navigational atom not in the scope of $\exists$ by $\top$ if it is implied by $\theta$ and
by $\bot$ in the opposite case. Note that $\theta(x,y) \wedge \psi_i(x,y)$ is equivalent to $\theta(x,y) \wedge \psi'_i(x,y)$. 
Let us split $\psi'(x,y)$ into conjuncts with free variable $x$ and conjuncts with free variable $y$: $\psi'_i(x) = \psi''_i(x) \wedge \psi''_i(y)$. 
We can write $\exists^{\ge k} y (\theta(x,y) \wedge \psi_i(x,y))$ equivalently as $\psi''_i(x) \wedge \exists^{\ge k} y ( \theta(x,y) \wedge \psi''_i(y))$. 
Finally, our translation depends on $\theta$. If $\theta  \in \{  \tprecv, \tsucch, \tprech, \teq \}$ then by the definition of $\cal G$ we have $k=0$ or $k=1$, so the
formula can be respectively replaced by $\top$ or $\exists y (\theta(x,y) \wedge \psi_i''(y))$. All the remaining cases can be treated as follows,
using Lemma \ref{l:definitions}:

\begin{itemize}\itemsep0pt
\item $\exists^{\ge k} y ( x \succv y \wedge \psi''_i(y))$ $\equiv$ 
$\Psiprop{k}{\posfive}{\psi''}(x)$

\item $\exists^{\ge k} y (  x \lessv y \wedge \neg (x \succv y) \wedge \psi''_i(y))$ $\equiv$ 
$\Psiprop{k}{\possix}{\psi''}(x)$

\item $\exists^{\ge k} y ( y \lessv x \wedge \neg (y \succv x) \wedge \psi''_i(y))$ 
$\equiv$ 
$\Psiprop{k}{\posseven}{\psi''}(x)$

\item $\exists^{\ge k} y (  x \lessh y \wedge \neg (x \succh y) \wedge \psi''_i(y))$ 
$\equiv$ 
$\Psiprop{k}{\posten}{\psi''}(x)$
    
\item $\exists^{\ge k} y ( y \lessh x \wedge \neg (y \succh x) \wedge \psi''_i(y))$ 
$\equiv$ 
$\Psiprop{k}{\poseleven}{\psi''}(x)$

\item $\exists^{\ge k} y ( x \fw y \wedge \psi''_i(y)) $
$\equiv$ 
$\hspace{-10pt}\bigvee\limits_{s{+}t{+}u{=}k}\hspace{-10pt}  (\Psiprop{s}{\postwelve}{\psi''}(x) \wedge \Psiprop{t}{\posfour}{\psi''}(x) \wedge \Psiprop{u}{\posfifteen}{\psi''}(x)).$
\end{itemize}
This finishes the process of replacing in $\phi$ a subformula starting with $\exists^{\ge k}$ by
an equivalent \FOt{} subformula. We proceed analogously with the remaining such subformulas, moving from the
deepest to the shallowest ones, and eventually obtain the desired formula $trans(\phi)$ without counting quantifiers.
\end{proof}
%%% Local Variables: 
%%% mode: latex
%%% TeX-master: "../main"
%%% End: 

%% file: conclusion.tex
\section{Combining the two extensions}
We proved that two extensions of two-variable logic on trees: the
extension $\CTwoFull$ with counting quantifiers, and the extension
\FOtall{} with additional uninterpreted binary relations remain
decidable and retain \ExpSpace-complexity of \FOt$[\succh, \lessh,
\succv, \lessv]$. It is tempting to combine both variants into a
single logic, i.e., to consider \Ctall , the two-variable logic with
counting quantifiers and additional binary relation over
trees. However, it turns out to lead to a very difficult
formalism. Namely, we can reduce to it the long standing open problem
of checking non-emptiness of vector addition tree automata.

\begin{theorem}
  The satisfiability problem for \Ctall{} is at least as hard as
  checking non-emptiness of vector addition tree automata.
\end{theorem}
\begin{proof}
  To prove the theorem we can mimic the reduction of vector addition
  tree automata to two-variable logic on \emph{data trees} given in
  Thm. 4.1 in \cite{BojanczykJACM09}.  Data trees are just trees with an
  additional, uninterpreted equivalence relation on nodes. In the
  reduction there the intended equivalence classes are of size at most
  two.  We can easily simulate this by a use a common binary symbol $E
  \in \tau_{com}$, constraining it to be reflexive and symmetric
  (which is naturally expressible in \FOt{}), and using counting
  quantifiers to enforce that each element is connected by $E$ to at
  most one other element. The remaining details of the proof remain
  unchanged. In the proof we do not need to use $\succh$ nor $\lessh$.
\end{proof}